\documentclass[sigconf, nonacm]{acmart}

\newcommand\vldbavailabilityurl{https://github.com/jeshi96/parallel-correlation-clustering}

\usepackage{bbm}
\usepackage{graphicx}
\usepackage{booktabs}
\usepackage{amsthm}
\usepackage{amsmath}
\usepackage{times}
\usepackage{multirow}
\usepackage{tabularx}
\usepackage{todonotes}
\usepackage{xspace}
\usepackage{algorithm}
\usepackage[noend]{algpseudocode}
\usepackage{subcaption,graphicx,url,multicol,mathtools}
\usepackage{paralist}
\usepackage{color}
\usepackage{arydshln}
\usepackage[most]{tcolorbox}
\usepackage{stmaryrd}
\usepackage{bm}
\newcommand{\myparagraph}[1]{\smallskip\noindent {\bf #1.}}

\newcommand{\id}[1]{\ifmmode\mathit{#1}\else\textit{#1}\fi}
\newcommand{\const}[1]{\ifmmode\mbox{\textc{#1}}\else\textsc{#1}\fi}

\newcommand{\unionnoarg}[1]{\ensuremath{\textsc{Union}}}
\newcommand{\uniontext}[1]{union}

\newcommand{\nghheap}[1]{neighbor-heap}

\newcommand{\rescaledweight}{rescaled weight}
\newcommand{\ccobj}{\textsc{CC}}
\newcommand{\pcomplete}{\ensuremath{\mathsf{P}}-complete}

\newcommand{\slfrac}[2]{\left.#1\middle/#2\right.}

\newcommand{\whp}[1]{w.h.p.}

\algblock{ParFor}{EndParFor}
\algnewcommand\algorithmicparfor{\textbf{parfor}}
\algnewcommand\algorithmicpardo{\textbf{do}}
\algnewcommand\algorithmicendparfor{}
\algrenewtext{ParFor}[1]{\algorithmicparfor\ #1\ \algorithmicpardo}
\algrenewtext{EndParFor}{\algorithmicendparfor}
\algtext*{EndParFor}

\algdef{SE}[DOWHILE]{Do}{doWhile}{\algorithmicdo}[1]{\algorithmicwhile\ #1}%

\makeatletter
\newtheorem*{rep@theorem}{\rep@title}
\newcommand{\newreptheorem}[2]{%
\newenvironment{rep#1}[1]{%
 \def\rep@title{#2 \ref{##1}}%
 \begin{rep@theorem}}%
 {\end{rep@theorem}}}
\makeatother

\newreptheorem{theorem}{Theorem}

\newcommand{\intheapp}{appendix\xspace}

\usepackage{xcolor}
\definecolor{lightred}{RGB}{252, 172, 167}
\definecolor{lightblue}{RGB}{207, 233, 252}
\definecolor{lightgreen}{RGB}{176, 252, 167}
\makeatletter
\newcommand{\algcolor}[2]{%
  \hskip-\ALG@thistlm\colorbox{#1}{\parbox{\dimexpr\linewidth-4\fboxsep}{\hskip\ALG@thistlm\relax #2}}%
}

\newcommand{\algemphblue}[1]{\algcolor{lightblue}{#1}}

\makeatother

\definecolor{change}{rgb}{0,0,0}
\definecolor{change2}{rgb}{0,0,0}

\begin{document}
\title{Scalable Community Detection via Parallel Correlation Clustering}

\author{Jessica Shi}
\affiliation{%
  \institution{MIT}
}
\email{jeshi@mit.edu}

\author{Laxman Dhulipala}
\affiliation{%
  \institution{MIT}
}
\email{laxman@mit.edu}

\author{David Eisenstat}
\affiliation{%
  \institution{Google}
}
\email{eisen@google.com}

\author{Jakub {\L{}}{ą}cki}
\affiliation{%
  \institution{Google}
}
\email{jlacki@google.com}

\author{Vahab Mirrokni}
\affiliation{%
  \institution{Google}
}
\email{mirrokni@google.com}

\begin{abstract}
Graph clustering and community detection are central problems
in modern data mining. The increasing need for analyzing billion-scale data  calls for faster and more scalable algorithms for these problems.
There are certain trade-offs between the
quality and speed of such clustering algorithms. In this paper, we design scalable algorithms that achieve high quality when evaluated based
on ground truth. 

We develop a generalized sequential and shared-memory parallel framework based on the \textsc{LambdaCC} objective (introduced by Veldt et al.), which encompasses modularity and correlation clustering. Our framework consists of highly-optimized implementations that 
scale to large data sets of billions of edges and that obtain high-quality clusters compared to ground-truth data, on both unweighted and weighted graphs. Our empirical evaluation shows that this framework improves the state-of-the-art trade-offs between speed and quality of scalable community detection. For example, on a 30-core machine with two-way hyper-threading, our implementations achieve orders of magnitude speedups over other correlation clustering baselines, and up to 28.44x speedups over our own sequential baselines while maintaining or improving quality. 
\end{abstract}

\maketitle

\section{Introduction}
As a fundamental tool in modern data mining, graph clustering, or community detection, has a wide range of applications spanning data mining \cite{JeBaPoMuMa15}, social network analysis \cite{GaLuMiYo11}, bioinformatics \cite{JiWaXiWa09}, and machine learning \cite{JaMuFl99}, and has been well-studied under many frameworks \cite{Schaeffer07, AgWa10}. As the need to analyze larger and larger data sets increases, designing scalable algorithms that can handle graphs with billions of edges has become a central part of graph clustering.
A major challenge is to design algorithms that can achieve fast speed at high scale while retaining high quality as evaluated on data sets with ground truth. Many graph clustering algorithms have been proposed to address this challenge, and our goal is to develop a state-of-the-art algorithm from both speed and quality perspectives. In particular, we adopt a new \textsc{LambdaCC} framework, introduced by  Veldt \textit{et al.} \cite{VeGlWi18}, which provides a general objective encompassing modularity~\cite{GiNe02} and correlation clustering~\cite{BaBlCh04}. Veldt \textit{et al.} show that \textsc{LambdaCC} framework unifies several quality measures, including modularity, sparsest cut, cluster deletion, and a general version of correlation clustering. Modularity is a widely-used objective that is formally defined as the fraction of edges within clusters minus the expected fraction of edges within clusters, assuming random distribution of edges. The goal of correlation clustering is to maximize agreements or minimize disagreements, where agreements and disagreements are defined based on edge weights indicating similarity and dissimilarity. 

It is NP-hard to approximate modularity within a constant factor \cite{DiLiTh15}, so optimizing for modularity, and by extension optimizing for the \textsc{LambdaCC} objective, is inherently difficult. The most successful and widely-used modularity clustering implementations focus on heuristic algorithms, notably the popular Louvain method \cite{BlGuLaLe08}. Indeed, the Louvain method has been well-studied for use in modularity clustering, with highly optimized heuristics and parallelizations that allow them to scale to large real-world networks \cite{LuHaKa15,Traag15, networkit, TrWaVa19}.

In this paper, we design, implement, and evaluate a generalized sequential and shared-memory parallel framework for Louvain-based algorithms including modularity and correlation clustering. In particular, we optimize the \textsc{LambdaCC} objective with state-of-the-art empirical performance, scaling to graphs with billions of edges. We also show that there is an inherent bottleneck to efficiently parallelizing the Louvain method, in that the problem of obtaining a clustering matching that given by the Louvain method on the \textsc{LambdaCC} objective, is \pcomplete{}. As such, we explore heuristic optimizations and relaxations of the Louvain method, and demonstrate their quality and performance trade-offs for the \textsc{LambdaCC} objective.

As part of our comprehensive empirical study, we show that our sequential implementation is orders of magnitude faster than the proof-of-concept implementation of Veldt \textit{et al.} \cite{VeGlWi18}.
We note that for both \textsc{LabmdaCC} and correlation clustering objective, we are unaware of any existing implementation that would scale to even million-edge graphs and achieve comparable quality.
We further show that our parallel implementations obtain up to 28.44x speedups over our sequential baselines on a 30-core machine.

Moreover, we show that optimizing for the correlation clustering objective is of particular importance, by studying cluster quality with respect to ground truth data.
We observe that optimizing for correlation clustering yields higher quality clusters than the ones obtained by optimizing for the celebrated modularity objective.
In addition, we compare our implementation to two other prominent scalable algorithms for community detection: Tectonic~\cite{TsPaMi17} and SCD~\cite{SCD} and in both cases obtain favorable results, improving both the performance and quality.
Finally, even in the highly competitive and extensively studied area of optimizing for modularity, we obtain  an up to $3.5$x speedup over a highly optimized parallel shared-memory modularity clustering implementation in NetworKit \cite{networkit}.

Our code is available at \url{\vldbavailabilityurl}.

\myparagraph{Further related work}
Optimization for correlation clustering has been studied empirically in the case of complete graphs, which is equivalent to \textsc{LambdaCC} objective with resolution $\gamma = 0.5$~\cite{elsner2009bounding, PaPaOyReRaJo15}.
{\color{change} In this restricted setting, several scalable parallel implementations have been obtained based on the \textsc{KwikCluster} algorithm~\cite{ChDaKu14, PaPaOyReRaJo15, GaKuBoTs20}. We observe that the \textsc{KwikCluster} algorithm typically obtains a negative \textsc{LambdaCC} objective, which significantly limits its practical applicability.}

Scalable modularity clustering has been extensively studied both in the shared-memory~\cite{networkit, FaMoMa17, halappanavar2017scalable, LuHaKa15, FENDER20171793, 7307558} and distributed memory~\cite{SaAr18, 7161493, riedy2012scalable, 8425242} settings.
The two fastest implementations that we identify are NetworKit~\cite{networkit} and Grappolo~\cite{halappanavar2017scalable, 8425242}.
Both of them offer comparable performance, but we observed the NetworKit typically computes solutions with slightly larger objective, and thus we compare to NetworKit in our empirical evaluation.
We also note that compared to these papers, our algorithm optimizes for a more general \textsc{LambdaCC} objective.

\section{Preliminaries} \label{sec:prelim}
We consider undirected weighted graphs $G = (V, E, w)$, where $w : E \rightarrow \mathbb{R}$ denotes the weight of each edge, and undirected unweighted graphs $G = (V, E)$, where $w_{uv} = 1$ for all $(u, v) \in E(G)$. We let $n = |V|$ and $m = |E|$, and we use $d_v$ to denote the degree of vertex $v$.

We use a generalized correlation clustering objective that is equivalent to the \textsc{LambdaCC} objective given by Veldt \textit{et al.} \cite{VeGlWi18}. Note that under a specific set of parameters, our objective similarly reduces to the classic modularity objective. Moreover, our definition can be more generally applied to weighted graph inputs.

We fix a clustering resolution parameter $\lambda \in (0, 1)$. We define non-negative \emph{vertex weights} $k : V \rightarrow \mathbb{R}^+_0$, where unless otherwise specified, we take $k_v = 1$ for all $v \in V$ (a redefinition of $k$ is required for the modularity objective). We also define the \emph{\rescaledweight} $w'$ of each pair of vertices $(u, v) \in V \times V$ to be $w'_{uv} = 0$ if $u = v$, $w'_{uv} = w_{uv} - \lambda k_u k_v$ if $(u, v) \in E$, and $w'_{uv} =  - \lambda k_u k_v$ otherwise.

The goal is to maximize the \ccobj{} objective,
$\ccobj(x) = \allowbreak \sum_{(i, j) \in V \times V} \allowbreak w'_{ij} \cdot \allowbreak (1-x_{ij})$,
where $x = \{x_{ij}\}$ represents the distance between vertices $i$ and $j$ in a given clustering. Specifically, $x_{ij} = 0$ if $i$ and $j$ are in the same cluster, and $x_{ij} = 1$ if $i$ and $j$ are in different clusters.

The modularity objective can be obtained from the \ccobj{} objective by defining vertex weights $k$ and setting $\lambda$ appropriately. Note that Reichardt and Bornholdt \cite{ReBo06} defined a modularity objective with a fixed scaling parameter $\gamma \in (0, 1)$ to be
$Q(x) = \frac{1}{2m} \sum_{i \neq j} (A_{ij} - \gamma \frac{d_i d_j}{2m}) (1 - x_{ij})$,
where $A_{ij} = 1$ if $i$ and $j$ are adjacent, and $A_{ij} = 0$ otherwise. Setting $\gamma = 1$, this objective is equivalent to the simpler modularity objective given by Girvan and Newman \cite{GiNe02}.
To modify \ccobj{} to match the modularity objective, we set the node weights $k(v) = d_v$ for each $v \in V$, and we set the resolution $\lambda = \slfrac{\gamma}{(2m)}$. Maximizing the two objective functions is then equivalent.

\section{Algorithm and Optimizations}

\subsection{Sequential Louvain Method}\label{sec:alg:seq-louvain}
We begin by describing 
the classic sequential Louvain method from Blondel \textit{et al.} \cite{BlGuLaLe08}, \textsc{Sequential-CC}, adapted for the correlation clustering objective. We include the pseudocode in the \intheapp{}. The main idea is to repeatedly move vertices to  clusters that would maximize the objective, and once no vertices can be moved, compress clusters into vertices and repeat this process on the compressed graph.
In more detail, the algorithm takes as input a graph $G$ and node weights $k$, and begins with singleton clusters.
Then, it iterates over each vertex in a random order, and locally moves vertices to clusters that maximize the \ccobj{} objective.
The computation to determine the cluster that vertex $v$ should move to can be performed by maintaining the total vertex weight $K_c$ of each cluster $c$ in $C$. We provide the formula in the \intheapp{}.

After all vertices have been moved, the algorithm repeats this step of locally moving vertices until no vertices have performed non-trivial moves.
If no vertices changed clusters during this phase, then \textsc{Sequential-CC} terminates.
Otherwise, once a stable state has been achieved, the algorithm compresses the graph $G$ (using a subroutine \textsc{Sequential-Compress})
by creating a new graph $G'$ with vertex weights $k'$.
Each cluster $c$ in $G$ corresponds to a vertex in $G'$ with vertex weight $k'(c) = K_c$. Edges $(u, v)$ in $G$ are maintained as edges between the vertices corresponding to their clusters in $G'$, where multiple edges incident on the same vertices are combined into a single edge with weight equal to the sum of their weights.

Finally, the algorithm recurses on $G'$ and $k'$.
The algorithm takes the returned clustering $C'$ on the compressed graph $G'$, and composes it with the original clustering $C$ (using a subroutine \textsc{Sequential-Flatten}). 
It assigns the cluster of a vertex $v$ in $G$ to be the cluster of its corresponding vertex in the compressed graph $G'$, composing the clustering obtained in the recursion onto the original graph.

The main bottleneck in parallelizing \textsc{Sequential-CC} is the sequential dependencies in moving each vertex to the cluster that maximizes the objective, 
and we prove a related \pcomplete{}ness result about this bottleneck in the \intheapp{} (showing that the problem of obtaining a clustering equivalent to any clustering $C$ given by moving each vertex to its best cluster, is inherently sequential to solve under standard complexity-theory assumptions).

\begin{figure}[t!]
\begin{center}
\includegraphics[width=0.16\columnwidth,page=1]{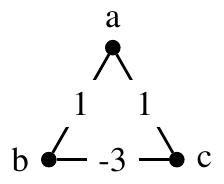}
\caption{An example where parallel local vertex moves lowers the total objective. Assume $\lambda = 0$ and initial clusters are singletons with objective 0. If $b$ and $c$ are both scheduled to move at the same time, they each choose cluster $\{a\}$, leading to a single cluster $\{a, b, c\}$ with objective -1.}
\label{fig:example-par-moves}
\end{center}
\end{figure}

As such, to obtain an empirically efficient implementation that achieves good parallelism, we heuristically relax the sequential dependency and allow vertices to move to clusters \emph{concurrently}. While vertices move to clusters that would individually maximize the objective, these moves in tandem may give a lower total objective; there is no guarantee of convergence. We show an example of this in Figure \ref{fig:example-par-moves}. Note that this is a common parallelization technique in Louvain methods for modularity clustering, and it has been observed that in practice this technique converges for the modularity objective \cite{networkit}.
We now discuss optimizations that can be used with this relaxation.

\subsection{Parallel Louvain Method and Optimizations} \label{sec:opt}

\begin{algorithm}[!t]
\footnotesize
\caption{Parallel Louvain method for correlation clustering}
 \begin{algorithmic}[1]
 \Procedure{Best-Moves}{$G$, $k$, $num\_iter$, $C$}
\State Define $\text{id}(c)$ to be the index of cluster $c \in C$
\State $D \leftarrow$ array of size $n$
\State $V' \leftarrow V$
\For{$i$ in range($num\_iter$)}
\ParFor{$v \in V'$} \label{pc:line-parfor-move}
\State $D[v] \leftarrow$ $\text{id}(c)$ such that moving $v$ to $c \in C$ would maximize \ccobj{($C$)} \label{pc:line-compute-move}
\State \algemphblue{Move $v$ to $D[v]$} \label{pc:line-async}
\EndParFor
\If{no moves were made in iteration $i$} break \EndIf
\State \algemphblue{$V' \leftarrow \{$ neighbors of $v \mid v$ moved$\}$  } \label{pc:line-move}
\EndFor
\State \Return $C$
 \EndProcedure
 \end{algorithmic}
\smallskip
\begin{algorithmic}[1]
\Procedure{Parallel-CC}{$G$, $k$, $num\_iter$}
\State $C \leftarrow$ singleton clusters for each $v \in V$
\State $C\leftarrow $ \textsc{Best-Moves}($G$, $k$, $num\_iter$, $C$)
\If{no moves were made in \textsc{Best-Moves}}
\State \Return $C$
\EndIf
\State $G', k' \leftarrow$ \textsc{Parallel-Compress}($G$, $C$)
\State $C' \leftarrow$ \textsc{Parallel-CC}($G'$, $k'$, $num\_iter$)
\State $C \leftarrow $ \textsc{Parallel-Flatten}($C$, $C'$)
\State \algemphblue{$C \leftarrow$ \textsc{Best-Moves}($G$, $k$, $num\_iter$, $C$)} \label{pc:line-refine}
\State \Return $C$
\EndProcedure
\end{algorithmic}
\label{alg:parallel-louvain}
\end{algorithm}

Algorithm \ref{alg:parallel-louvain} contains the pseudocode for each of the main optimizations that we consider in our parallelization of the Louvain method for the \ccobj{} objective. Each optimization is highlighted in blue, and we display in the pseudocode only the best settings that offer a reasonable trade-off between quality and performance, as we show in Section \ref{sec:exp:tuning}. We discuss in the following sections the other modular options available in our framework for each of these optimizations.

Our parallelization uses a natural heuristic relaxation of the sequential dependency in moving vertices to their desired clusters. A more faithful parallelization would fix a random permutation of $V$, and move in parallel the first $\ell$ vertices in order for the largest $\ell$ such that moving these $\ell$ vertices would not affect each other's objectives. However, compared to a heuristic relaxation, not only does this involve greater overhead due to the prefix computation of vertices that do not conflict, but it also respects sequential dependencies that may not affect later vertex moves. Thus, in the interest of performance, we consider the parallelization given in Algorithm \ref{alg:parallel-louvain}.

Note that the heuristic relaxation to allow vertices to move concurrently to their desired clusters is encapsulated in the subroutine \textsc{Best-Moves}. We include an additional parameter, $num\_iter$, which bounds the number of iterations in which we move each vertex to its desired cluster; this is necessary due to the lack of guarantee of convergence. Moreover, in order to allow each vertex to efficiently compute its best move, we maintain the total vertex weights $K_c$ of each cluster $c$, which is not shown in the pseudocode for simplicity.

{\color{change} Our main optimizations introduce symmetry breaking and work reduction techniques to improve performance while maintaining the objective. We also discuss a refinement step that improves the objective at the cost of running time and a higher memory overhead.}

\subsubsection{Optimization: Synchronous vs Asynchronous} \label{sec:opt:sync}
The first optimization involves scheduling individual vertex moves in \textsc{Best-Moves}, on Line \ref{pc:line-async}. We explore two options: \emph{synchronous} and \emph{asynchronous}. These options have been previously studied in parallelizing the Louvain method for the modularity objective \cite{networkit, SaAr18}.

In the synchronous setting, instead of moving vertices on Line \ref{pc:line-async} immediately after the computation of their desired cluster, we move all vertices in parallel to their desired cluster $D[v]$ after the parallel for loop on Line \ref{pc:line-parfor-move}. This can be efficiently performed in parallel by aggregating vertices that move from the same clusters and vertices that move to the same clusters.

In the asynchronous setting, we perform vertex moves on Line \ref{pc:line-async} as highlighted in blue. Note that moving a vertex $v$ in this manner potentially interferes with the computation that each vertex performs on Line \ref{pc:line-compute-move}, where other vertices' computations of their desired cluster may depend on $v$'s current cluster, and the total vertex weight of $v$'s prior and $v$'s new cluster. Instead of using locks or other methods of synchronization, we relax consistency guarantees, and we perform separate atomic operations to update the cluster that $v$ moves to and to update the total vertex weight of the cluster that $v$ moves to. Thus, there is no guarantee that the stored total vertex weights of clusters represent the actual total vertex weights of the clusters.

We show in Section \ref{sec:exp:tuning} that perhaps surprisingly, the asynchronous setting outperforms the synchronous setting, particularly in terms of objective. {\color{change2} Of the optimizations, the asynchronous optimization contributes most significantly towards an improvement in objective.} This is because the asynchronous setting allows for symmetry breaking, whereas in the synchronous setting, vertices that are attempting to move away from each other may inadvertently move to the same cluster, since they must move in lockstep. For certain graphs, this symmetry breaking also allows the asynchronous setting to outperform the synchronous setting in terms of running time, due to fewer vertex moves required to obtain the maximal objective.

Previous uses of the Louvain method for other objectives explored different schedules for vertex moves, which give more granular trade-offs~\cite{LuHaKa15, BaHaWeRoHo17}. We found that our asynchronous setting outperforms methods that maintain consistency guarantees in quality and speed.

\subsubsection{Optimization: All Vertices vs Neighbors of Clusters vs Neighbors of Vertices}
We now consider optimizations that reduce the set of vertices to consider moving in every iteration of \textsc{Best-Moves}. When considering vertices on Lines \ref{pc:line-parfor-move} -- \ref{pc:line-compute-move}, we note that following a set of vertex moves in the previous iteration, we can reduce the number of vertices that would be likely to be induced to change clusters by the vertex moves in the previous iteration. This idea has been previously used in work on the Louvain method for the modularity objective \cite{OzTeIn16, BaHaWeRoHo17}. Importantly, the \textsc{Best-Moves} subroutine takes a significant portion of total clustering time, and reducing the subset of vertices to consider offers significant performance improvements.

In more detail, isolating a vertex $v$ which has in the previous iteration moved from cluster $c$ to cluster $c'$, the vertices that would be affected by this move in the next iteration belong to three categories: (a) neighbors of $v$, {\color{change} (b) neighbors of any vertex in $c$ , and (c) vertices in $c'$.} Any vertex $u$ that is neither a neighbor of $v$ nor a neighbor of $c$, and that is not in $c'$, is not induced to move clusters due to $v$'s move. This is due to the change in objective formula.
{\color{change} For conciseness in describing category (b), we formally define \emph{neighbors of clusters} $C$ to be the union of the neighbors of each vertex in each cluster in $C$.}

In our algorithm, we consider three options for this optimization: restricting considered vertices to \emph{neighbors of vertices} moved in the previous iteration, restricting considered vertices to \emph{neighbors of clusters} that vertices have moved to in the previous iteration, and considering \emph{all vertices} in each iteration. 
The first option corresponds to the update on $V'$ in Line \ref{pc:line-move}, highlighted in blue. The second option would instead replace this line with setting $V' \leftarrow \{$ neighbors of the current cluster $C[i] \mid v$ moved from cluster $C[i]$ to cluster $c\}$, and the final option would set $V' \leftarrow V$.

We show in Section \ref{sec:exp:tuning} that restricting $V'$ to the set of neighbors of vertices that have moved in the previous iteration outperforms both other options while maintaining comparable objective. This is because the vertices that are most affected by moving vertices in terms of objective are neighbors of the moving vertices, and thus most of the objective obtained is from considering these neighbors.
While there may be contrived scenarios in which more objective is affected due to non-neighbors of moving vertices with sufficiently large edge weights \footnote{\color{change2}For instance, given a large enough star graph where each leaf has a small enough positive edge weight to the center, following the first set of moves in which every leaf clusters with the center.}, we do not see these scenarios in practice, and 
considering a smaller subset of vertices in each iteration allows for less total work to be performed, since we save on the cost of computing best moves for other vertices. The performance improvements outweigh the marginal loss in objective from these cases.

\subsubsection{Optimization: Multi-level Refinement} 
Finally, we consider a popular multi-level refinement optimization \cite{RoNo11, networkit}. Note that the first phase of our parallel algorithm and the classic Louvain method involves what can be viewed as successive coarsening steps, in which we perform best vertex moves and compress the resulting clustering into a coarsened graph; {\color{change} vertices in the coarsened graph correspond to clusters, or sets of vertices, in the original graph. For instance, each vertex $v$ in the original graph is clustered into a cluster $C$, which corresponds to the vertex $v'$ in the coarsened graph.
We then recurse on the coarsened graph. Following the recursion, we receive a clustering on the coarsened graph, which we must translate to a clustering on the original graph. Importantly, each vertex $v'$ in the coarsened graph now belongs to a cluster $C'$ in the coarsened graph, and we must now assign the cluster of the original vertex $v$. We use a flattening procedure for this, where we simply assign each vertex $v$ in the original graph to the corresponding cluster $C'$.}

{\color{change} However, note that $v$ did not have an opportunity to move clusters individually in successive recursive steps, and importantly, because there is no guarantee of convergence, clusters may not reach a steady state prior to compressing the graph. $v$ may have ended up in a sub-optimal cluster $C$ when the coarsening was performed and would have been unable to change clusters after the coarsening.
Now, given its new cluster $C'$, $v$ may desire to change clusters.}
{\color{change} The \emph{multi-level refinement} optimization allows for $v$ to move by performing a refinement step after each flattening step, as we traverse back up the recursive hierarchy.  We simply perform a further iteration of \textsc{Best-Moves} on each individual vertex $v$, prior to returning the clustering.}

This refinement optimization is shown on Line \ref{pc:line-refine}, highlighted in blue. Omitting this line removes the optimization. The optimization increases the space usage and the amount of time required for our implementation, since it requires each compressed graph to be maintained throughout and since it adds an additional subroutine, but it non-trivially improves quality, as we show in Section \ref{sec:exp:tuning}. {\color{change} This is precisely due to the lack of guarantee of convergence in Algorithm \ref{alg:parallel-louvain}, where vertices may be coarsened non-optimally. The refinement step allows for these vertices to move to better clusters as we traverse back up the recursive hierarchy, resulting in better overall quality.}

\subsubsection{Other Optimizations}\label{sec:opt:other-opt}
We discuss in the \intheapp{} other practical optimizations that we use, including our efficient parallelization of the subroutines \textsc{Sequential-Compress} and \textsc{Sequential-Flatten}.

\section{Experiments}
In this section, we present a comprehensive evaluation of our algorithms, demonstrating significant speedups over state-of-the-art implementations and high-quality clusters compared to ground truth.

{\color{change} We show that optimizations that address symmetry breaking and work reduction result in an overall faster implementation while maintaining objective, and additional refinement steps improve the objective at the cost of performance and memory usage. We also demonstrate significant speedups over state-of-the-art implementations due to our theoretically efficient parallelization of key subroutines, while obtaining high quality compared to ground truth.}

\myparagraph{Environment}
We run most experiments on a c2-standard-60 Google Cloud instance, with 30 cores (with two-way hyper-threading), 3.8GHz Intel Xeon Scalable processors, and 240 GiB main memory. For experiments on large graphs we use a m1-megamem-96 Google Cloud instance, with 48 cores (with two-way hyper-threading), 2.7GHz Intel Xeon Scalable processors, and 1434 GiB main memory. We compile our programs with g++ (version 7.3.1) and the \texttt{-O3} flag, and we use an efficient work-stealing scheduler, {\color{change} which, as shown in~\cite{BlAnDh20}, provides on average a 1.43x speedup over Intel's Parallel STL library.} We also terminate any experiment that takes over 7 hours.

\myparagraph{Graph Inputs}
We test our implementations on real-world undirected graphs from the Stanford Network Analysis Project (SNAP) \cite{SNAP}, namely com-dblp, com-amazon, com-livejournal, com-orkut, and com-friendster. We also use twitter, a symmetrized version of the Twitter graph representing follower-following relationships \cite{Kwak2010}. Details of these graphs are shown in the \intheapp{}.
To show the clustering quality of our implementations, we compare with the top 5000 ground-truth communities given by SNAP. These communities may overlap, so to compute average precision and recall, for each ground-truth community $c$, we match $c$ to the cluster $c'$ with the largest intersection to $c$.\footnote{Any given cluster $c'$ may be matched to multiple or no ground-truth communities $c$.} This matches the methodology used by Tsourakakis \textit{et al.} in evaluating \textsc{Tectonic} \cite{TsPaMi17}.

We also use an approximate $k$-NN algorithm~\cite{scann} to construct weighted graphs from pointset data, from the UCI Machine Learning repository \cite{DuGr17}. We defer a discussion of our weighted graph data to the \intheapp{}.
Additionally, we demonstrate scalability using synthetic graphs generated by the standard rMAT graph generator~\cite{ChakrabartiZF04}, with $a=0.5$, $b=c=0.1,$ and $d=0.3$.

All experiments are run on the c2-standard-60 instances, except for experiments on the twitter and friendster graphs, which are run on the m1-megamem-96 instances due to the higher memory requirement, particularly when using multi-level refinement.

\myparagraph{Implementations} We test the Louvain-based implementations of our sequential and parallel correlation clustering algorithms (\textsc{Seq-CC} and \textsc{Par-CC} respectively). We also redefine vertex weights and $\lambda$ as discussed in Section \ref{sec:prelim} to obtain modularity clustering implementations (\textsc{Seq-Mod} and \textsc{Par-Mod}). For our parallel implementations, we use $num\_iter = 10$ unless otherwise specified. For our sequential implementations, we use the superscript ${}^\textsc{con}$ if we run to convergence (without restricting the number of iterations), and we use no superscript if we use $num\_iter = 10$.
We run each experiment 10 times and report the average time and objective. \footnote{The average objective is non-deterministic when using the asynchronous setting from Section \ref{sec:opt:sync}.} 

We compare to two correlation clustering implementations, namely the parallel \textsc{C4} and \textsc{ClusterWild!} by Pan \textit{et al.} \cite{PaPaOyReRaJo15} (based on the sequential correlation clustering algorithm \textsc{KwikCluster} \cite{NiChNe08}), and the sequential Louvain-based implementation in the correlation clustering framework \textsc{LambdaCC}, by Veldt \textit{et al.} \cite{VeGlWi18}.

We note that there is a rich body of work on pivot-based correlation clustering algorithms, both parallel and sequential, including prior work by Chierichetti \textit{et al.} \cite{ChDaKu14} and by Garc\'{\i}a-Soriano \textit{et al.} \cite{GaKuBoTs20}. These works are based on \textsc{KwikCluster} (also known as \textsc{Pivot}), and offer faster performance with matching or worse approximation guarantees. However, in our comparison to \textsc{C4}, which parallelizes \textsc{KwikCluster} and which matches \textsc{KwikCluster}'s approximation guarantee, we note that while \textsc{C4} is much faster than \textsc{Par-CC}, the quality is poor compared to \textsc{Par-CC} {\color{change2} in terms of both the \ccobj{} objective and comparison to ground-truth communities, resulting in clusters of vertices with proportionally lower similarities to each other}. Thus, we omit pivot-based work that offer the same or worse quality guarantees compared to \textsc{C4}.

We also compare to two state-of-the-art community detection algorithms which are based on triangle counts: the sequential \textsc{Tectonic} by Tsourakakis \textit{et al.} \cite{TsPaMi17}, and the shared-memory parallel \textsc{SCD}, by Prat-P\'{e}rez \textit{et al.} \cite{SCD}.
Both algorithms were shown to deliver superior quality to multiple similarly scalable baseline methods.
In the special case of modularity, we compare to the parallel Louvian-based modularity clustering implementation in the NetworKit toolkit (\textsc{networkit}), by Staudt and Meyerhenke \cite{networkit}.

\begin{figure}[t!]
\begin{center}
\begin{subfigure}[t]{\columnwidth}
\centering
\includegraphics[width=\textwidth,page=15]{figures/fig-new.pdf}
\end{subfigure}
\hfill
\begin{subfigure}[t]{\columnwidth}
\includegraphics[width=\textwidth,page=1]{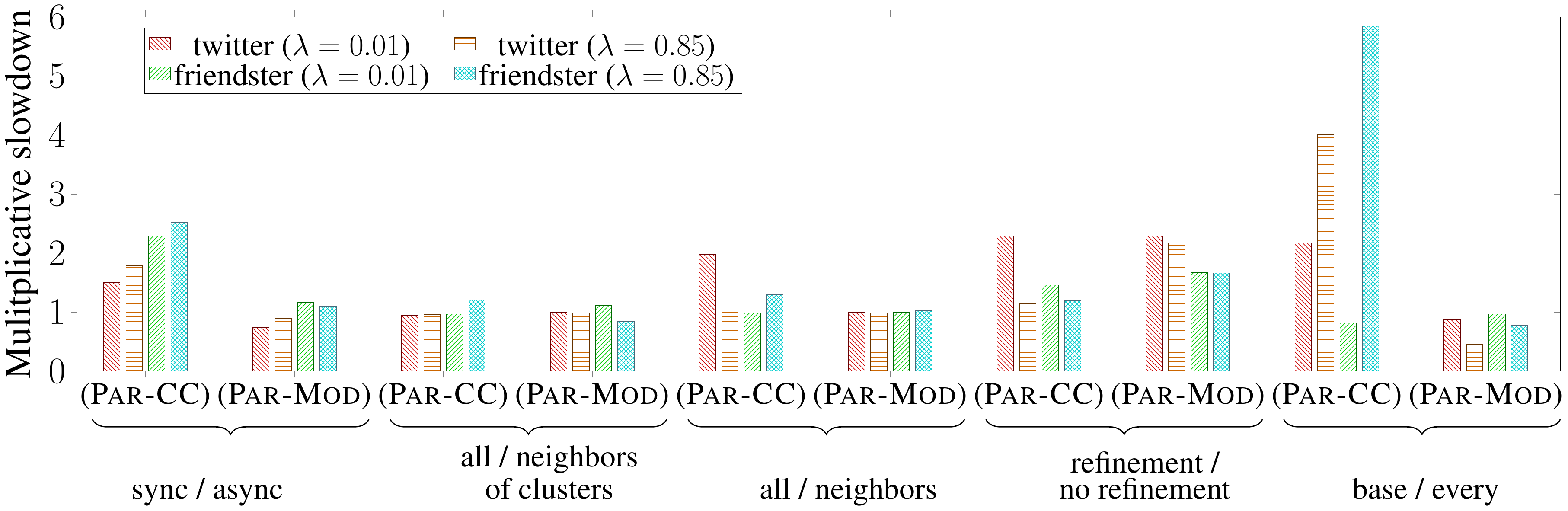}
\end{subfigure}
\caption{{\color{change} Multiplicative slowdown in average time of each optimization. Also shown is the slowdown for no optimizations (base) over every optimization. Running times were obtained for \textsc{Par-CC} and \textsc{Par-Mod} on the graphs amazon, orkut, twitter, and friendster, with $\lambda = 0.01, 0.85$.} }
\label{fig:exp1-time}
\end{center}
\end{figure}

\begin{figure}[t!]
\centering
\begin{subfigure}[t]{0.45\textwidth}
\centering
\includegraphics[width=\textwidth,page=5]{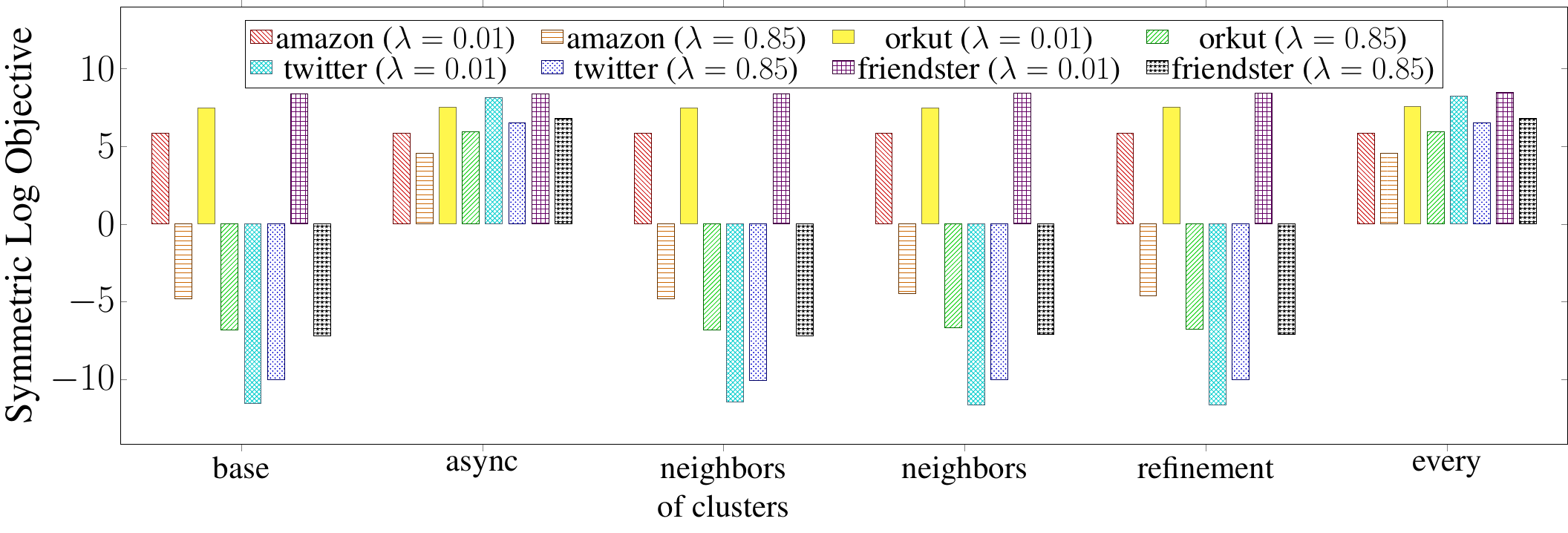}
\end{subfigure}
\hfill
\begin{subfigure}[t]{0.35\textwidth}
\centering
\includegraphics[width=\textwidth,page=6]{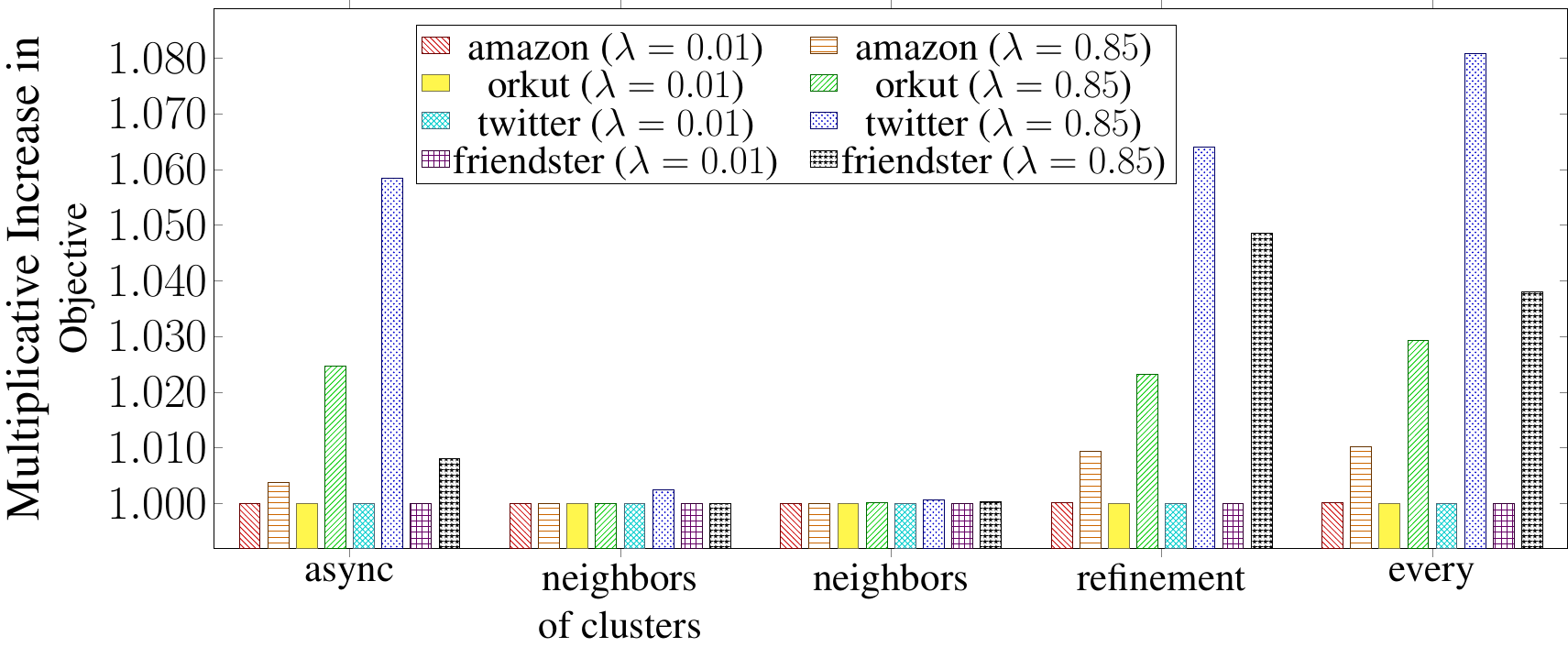}
\end{subfigure}
\caption{{\color{change} \ccobj{} objective for \textsc{Par-CC} on a symmetric log scale\protect\footnotemark (top) and multiplicative increase in the modularity for \textsc{Par-Mod} over no optimizations (bottom), of each optimization and every optimization. Objectives were obtained on amazon, orkut, twitter, and friendster, with $\lambda = 0.01, 0.85$.}}
\label{fig:exp1-obj}
\end{figure}
\footnotetext{{\color{change} We take the symmetric log of $x$ to be $\text{sign}(x) \cdot\log{|x|}$. We use a symmetric log scale to more accurately depict the \ccobj{} objectives, because the objectives are very large positive and negative numbers.}}

\subsection{Tuning Optimizations} \label{sec:exp:tuning}
We evaluated the effectiveness of the different optimizations discussed in Section \ref{sec:opt}, namely, considering \emph{synchronous} versus \emph{asynchronous} vertex moves, considering \emph{all vertices} versus \emph{neighbors of clusters} that vertices have moved to versus \emph{neighbors of vertices} that have moved as the vertex subset $V'$ to iterate over, and considering \emph{multi-level refinement} versus \emph{no refinement}. We establish here that the optimizations that offer reasonable trade-offs between speed and quality are asynchronous vertex moves, considering neighbors of vertices that have moved as $V'$, and using multi-level refinement.

{\color{change} We tuned these optimizations on the graphs amazon, orkut, twitter, and friendster, with $\lambda = 0.01, 0.85$. Note that lower resolutions produce a clustering with fewer clusters, and higher resolutions produce a clustering with more clusters; these resolutions effectively model both scenarios, where differences in the number of clusters produced may affect performance. We fix the synchronous, all vertex moves, and no refinement options, 
and give running times and objectives considering turning on a single optimization at a time; these are the natural settings that do not optimize the basic sequential Louvain method. 
Considering both \textsc{Par-CC} and \textsc{Par-Mod}, Figure \ref{fig:exp1-time} shows the multiplicative slowdowns of synchronous over asynchronous, all vertices over neighbors of clusters, all vertices over neighbors of vertices, multi-level refinement over no refinement (note that multi-level refinement causes slowdowns, but improves quality over the basic no refinement option), and no optimizations over every optimization.
Figure \ref{fig:exp1-obj} shows the objectives for these optimizations.}

We see speedups of up to 2.50x, {\color{change2}with a median of 1.21x,} using the asynchronous over the synchronous setting across \textsc{Par-CC} and \textsc{Par-Mod}. For \textsc{Par-CC}, the synchronous setting often produces negative objective, whereas in the asynchronous setting, the objective is always positive, and we see a 1.29--156.01\% increase in objective. The objective in the synchronous setting is negative likely due to the phenomenon shown in Figure \ref{fig:example-par-moves}, which is more likely to appear for large resolutions due to the objective computation. {\color{change2} This phenomenon has additionally been discussed in prior work in relation to the modularity objective \cite{RaAlKu07, networkit}.}
For \textsc{Par-Mod}, we see a 0.0015 -- 5.84\% increase in modularity using the asynchronous setting over the synchronous setting. The synchronous setting often leads to very poor objectives due to a lack of symmetry breaking, compared to the asynchronous setting where there is inherent randomness.

The asynchronous setting is also often faster than the synchronous setting, because due to the symmetry breaking, fewer vertices end up making moves that decrease the objective.
For \textsc{Par-Mod} on orkut and twitter, the asynchronous setting is not faster than the synchronous setting, but the increase in modularity using the asynchronous setting is more significant, compared to that obtained on other graphs. Up to 1.43x more time is spent computing best moves in the asynchronous setting compared to the synchronous setting, due to an increased number of vertex moves required to obtain the higher objective.
Overall, considering tradeoffs between objective and speed, the best setting in general is the asynchronous setting.

 We see up to a 1.98x speedup, {\color{change2}with a median of 1.03x,} considering neighbors of vertices compared to all vertices as the subset $V'$, and we see up to a 1.32x speedup, {\color{change2}with a median of 1.01x,} considering neighbors of clusters compared to all vertices. Moreover, the objectives obtained in all settings are comparable. This is because neighbors of previously moved vertices, and by extension neighbors of clusters of previously moved vertices, are most significantly affected in terms of objective by previously moved vertices, based on the change in objective formula. In the cases where the speedup is minimal, vertices in these classes represent a larger proportion of all vertices, which we show in the \intheapp{}; however, due to the cases with more significant speedups, the best setting in general is considering neighbors of vertices as the subset $V'$.

Finally, we see slowdowns of up to 2.29x, {\color{change2}with a median of 1.67x,} using multi-level refinement, compared to using no refinement. However, using multi-level refinement, we see 1.12 -- 36.92\% increase in the \ccobj{} objective, and up to a 6.41\% increase in modularity. Refinement improves objective because it allows vertices to move to better clusters following the compression steps, increasing the objective in situations where compression was not optimal. The increase in time is due to the added work in refinement, and in general, the best setting is to use refinement. {\color{change2} Note that these results mirror prior work applying multi-level refinement for the modularity objective \cite{RoNo11, networkit}.} For the modularity objective using small resolutions, the increase in objective is minimal; this is because in these cases, the objective obtained without using refinement is already very high (on average 0.99, where the maximum is 1.00).

In the remaining experiments, we fix the asynchronous setting, using neighbors of vertices, and using multi-level refinement, as the overall optimal settings, although we note that for small resolutions, multi-level refinement often offers little increase in objective considering the increase in running time. Overall, using all optimizations, we see up to a 5.85x speedup in running time and up to a 156.01\% increase in objective. For \textsc{Par-Mod}, there are scenarios with up to a 2.20x slowdown in running time due to contention in the asynchronous setting compared to the synchronous setting, but in these cases, we see significant increases in modularity, of 2.93\% -- 8.09\%.

\subsection{Speedups and Scalability}

\begin{figure}[t!]
\centering
\begin{subfigure}[t]{0.2\textwidth}
\centering
\includegraphics[width=\textwidth,page=1]{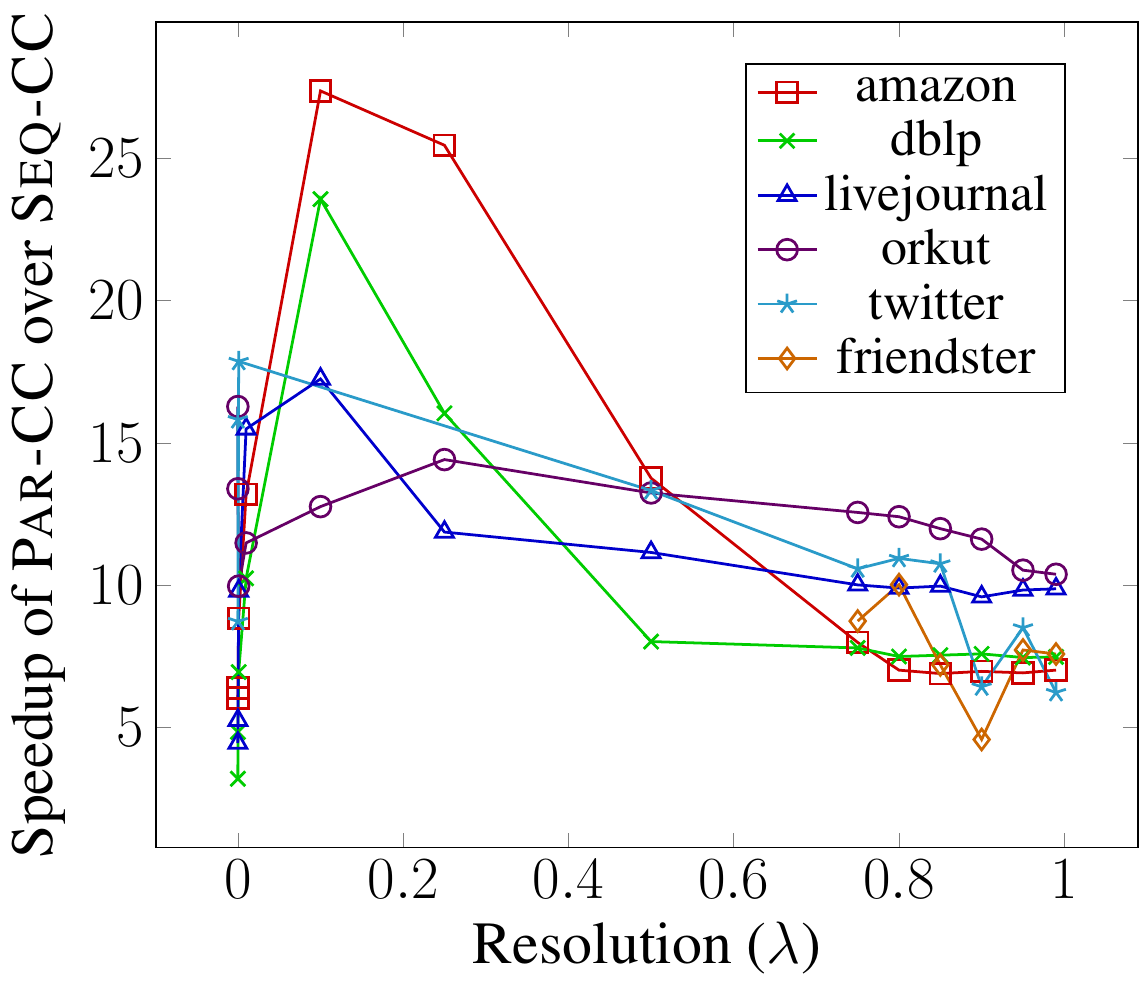}
\end{subfigure}
\hfill
\begin{subfigure}[t]{0.2\textwidth}
\centering
\includegraphics[width=\textwidth,page=2]{figures/fig-new.pdf}
\end{subfigure}
\caption{{\color{change} Speedup of \textsc{Par-CC} over \textsc{Seq-CC} (left) and of \textsc{Par-Mod} over \textsc{Seq-Mod} (right), on amazon, dblp, livejournal, orkut, twitter, and friendster, for varying resolutions. \textsc{Seq-CC} timed out on twitter for $\lambda = 0.01$, $0.1$, and $0.25$, and on friendster for $\lambda < 0.75$. }}
\label{fig:exp5-trunc}
\end{figure}

\begin{figure}[t!]
\centering
\begin{subfigure}[t]{0.2\textwidth}
\centering
\includegraphics[width=\textwidth,page=3]{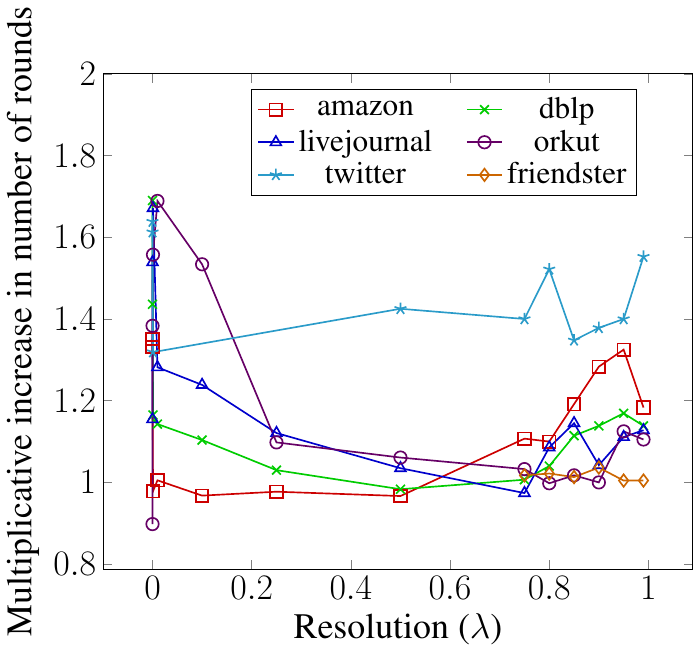}
\end{subfigure}
\hfill
\begin{subfigure}[t]{0.2\textwidth}
\centering
\includegraphics[width=\textwidth,page=4]{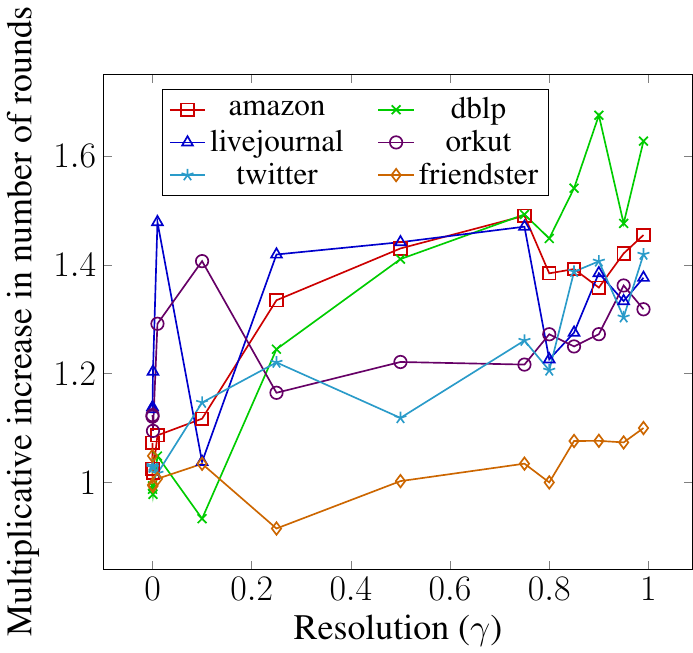}
\end{subfigure}
\caption{{\color{change} Multiplicative increase in the number of rounds until convergence or until the default number of iterations is reached, of \textsc{Par-CC} over \textsc{Seq-CC} (left) and of \textsc{Par-Mod} over \textsc{Seq-Mod} (right), on amazon, dblp, livejournal, orkut, twitter, and friendster, for varying resolutions. \textsc{Seq-CC} timed out on twitter for $\lambda = 0.01$, $0.1$, and $0.25$, and on friendster for $\lambda < 0.75$. }}
\label{fig:exp5-numiter}
\end{figure}

\myparagraph{Speedups} We first note that there exist no prior scalable correlation clustering baselines that offer high quality in terms of objective. The existing implementations are the parallel \textsc{C4} and \textsc{ClusterWild!} \cite{PaPaOyReRaJo15}, which are based on a maximal independent set algorithm, and the sequential Louvain-based method in \textsc{LambdaCC} \cite{VeGlWi18}. Our implementations significantly outperform these baselines, and we present a detailed comparison in the \intheapp{}. Notably, \textsc{C4} and \textsc{ClusterWild!} offer significant speedups of up to 428.64x over \textsc{Par-CC}, but 
achieve poor and often negative objective, with a decrease in the objective of 273.35 -- 433.31\% over \textsc{Par-CC}. {\color{change2} \textsc{C4} and \textsc{ClusterWild!} also achieve poor precision and recall compared to ground truth communities, with precision between 0.44 -- 0.65 and corresponding recall between 0.10 -- 0.15. In comparison, on the same graphs, \textsc{Par-CC} achieves recall between 0.61 -- 0.98 for precision greater than 0.50.}
Furthermore, \textsc{LambdaCC} is a MATLAB implementation that uses an adjacency matrix to represent the graph, and cannot scale to graphs of more than hundreds of vertices. 
Thus, we demonstrate the speedups of our parallel implementations primarily against our own sequential implementations, which include the applicable optimizations discussed in Section \ref{sec:exp:tuning}, namely considering only neighbors of vertices when computing best local moves, and multi-level refinement.

Figure \ref{fig:exp5-trunc} shows the speedup of \textsc{Par-CC} and \textsc{Par-Mod} over \textsc{Seq-CC} and \textsc{Seq-Mod} respectively.
We also compared to \textsc{Seq-CC}${}^\textsc{con}$ and \textsc{Seq-Mod}${}^\textsc{con}$; we note that running to convergence generally increases the running time while improving the objective, although the improvements are not always significant. However, as we show later in Section \ref{sec:exp:quality:uw}, the average precision-recall of \textsc{Seq-CC} is significantly worse than that of \textsc{Seq-CC}${}^\textsc{con}$, while our \textsc{Par-CC} matches the average precision-recall of \textsc{Seq-CC}${}^\textsc{con}$.

\begin{figure}[t!]
\centering
\begin{subfigure}[t]{0.22\textwidth}
\centering
\includegraphics[width=\textwidth,page=13]{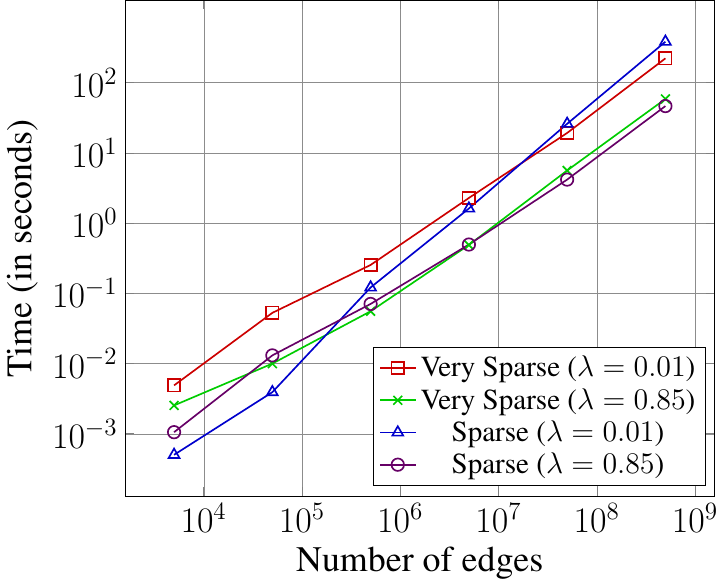}
\end{subfigure}
\hfill
\begin{subfigure}[t]{0.22\textwidth}
\centering
\includegraphics[width=\textwidth,page=8]{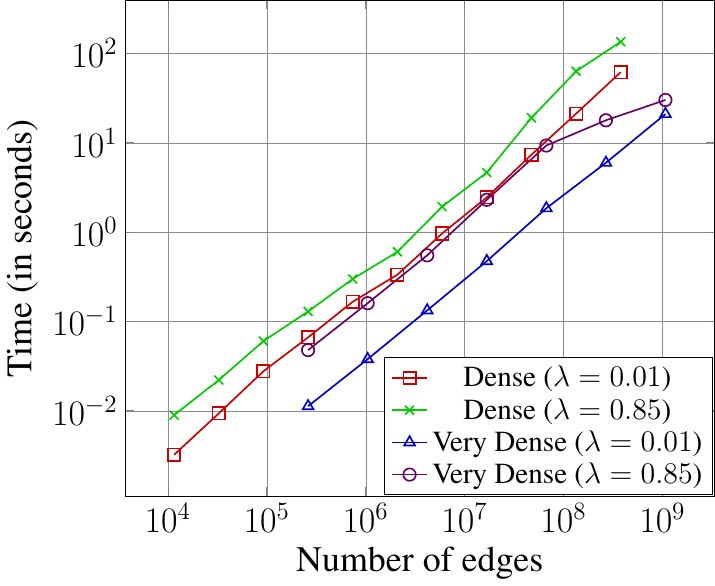}
\end{subfigure}
\caption{{\color{change} Scalability of \textsc{Par-CC} over rMAT graphs of varying sizes.}}
\label{fig:exp7}
\end{figure}

\begin{figure}[t!]
\centering
\begin{subfigure}[t]{0.2\textwidth}
\centering
\includegraphics[width=\textwidth,page=3]{figures/fig.pdf}
\end{subfigure}
\hfill
\begin{subfigure}[t]{0.2\textwidth}
\centering
\includegraphics[width=\textwidth,page=1]{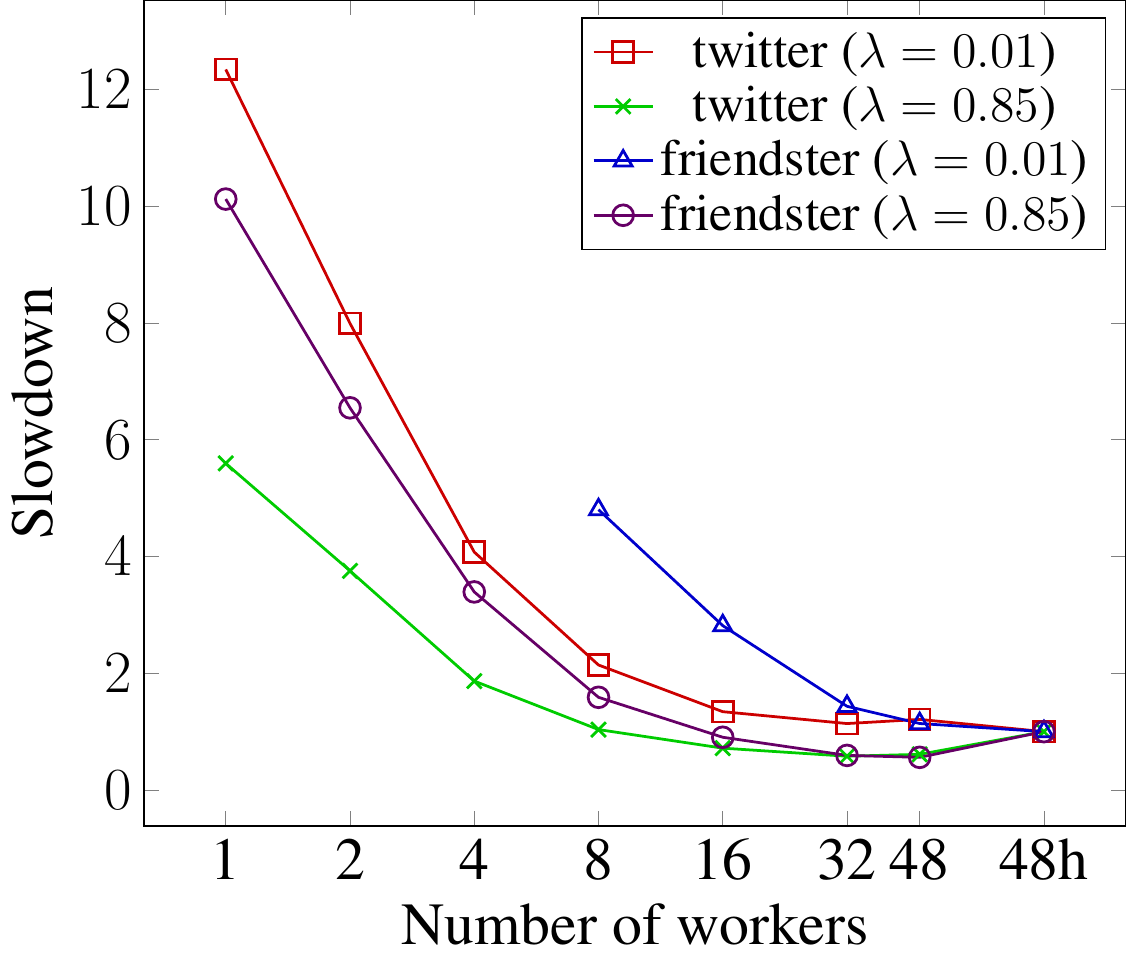}
\end{subfigure}
\caption{{\color{change} Scalability of \textsc{Par-CC} over different numbers of threads, on amazon, orkut, twitter, and friendster, with $\lambda = 0.01, 0.85$. 30h and 48h indicate 30 and 48 cores respectively, with two-way hyper-threading.}}
\label{fig:exp6}
\end{figure}

\begin{figure}[t!]
\centering
\includegraphics[width=0.35\textwidth,page=7]{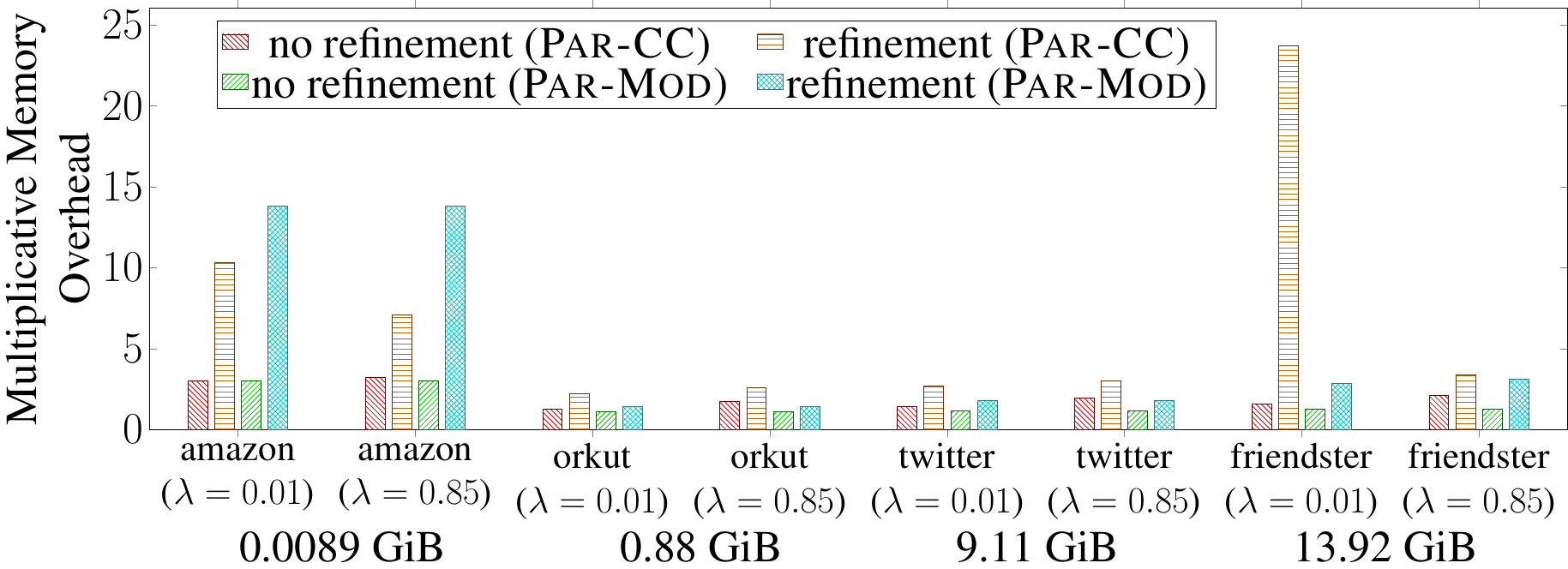}
\caption{{\color{change} Multiplicative memory overhead, over the size of the input graph, of \textsc{Par-CC} and \textsc{Par-Mod}, on amazon, orkut, twitter, and friendster, with $\lambda = 0.01, 0.85$. The labels below denote the size of the input graph.}}
\label{fig:expmem}
\end{figure}

On the graphs amazon, dblp, livejournal, and orkut, and over varying resolutions, we see 3.19--27.38x speedups of \textsc{Par-CC} over \textsc{Seq-CC}, and 12.55--110.25x speedups of \textsc{Par-CC} over \textsc{Seq-CC}${}^\textsc{con}$. Our parallel implementations also achieve between 0.98--1.08x the \ccobj{} objective of our serial implementations, demonstrating high performance while maintaining the \ccobj{} objective. For \textsc{Par-Mod}, we see between 3.18--7.76x speedups over \textsc{Seq-Mod}, and 2.64--7.89x speedups over \textsc{Seq-Mod}${}^\textsc{con}$, while achieving 1.00--1.06x the modularity of the serial implementations.

On the large graphs twitter and friendster, and over varying resolutions, we see 4.57--17.87x speedups of our \textsc{Par-CC} over \textsc{Seq-CC}, and up to 7.74x speedups of \textsc{Par-Mod} over \textsc{Seq-Mod}. We achieve between 0.95--1.00x the \ccobj{} objective of \textsc{Seq-CC}, and between 0.96--1.02x the modularity of \textsc{Seq-Mod}.
\textsc{Seq-CC} timed out on twitter for $\lambda = 0.01$, $0.1$, and $0.25$, and on friendster for $\lambda < 0.75$. We observe lower speedups using \textsc{Par-Mod} on twitter, which we discuss in the \intheapp{}.

Figure \ref{fig:exp5-numiter} shows the multiplicative increase in the total number of iterations required for \textsc{Par-CC} and \textsc{Par-Mod} over \textsc{Seq-CC} and \textsc{Seq-Mod}, which approximately displays the inverse of the behavior seen in the speedups shown in Figure \ref{fig:exp5-trunc} across different resolutions. We see greater speedups for resolutions where the number of iterations required in our parallel implementations match or are lower than the number of iterations required in our serial implementations.
Whenever a greater number of iterations is required in parallel compared to serial, we observed lower speedups, simply due to the increased amount of work carried out by the parallel version.

\myparagraph{Scalability} {\color{change} Figure \ref{fig:exp7} demonstrates the scalability of \textsc{Par-CC} over rMAT graphs of varying sizes, with very sparse graphs ($m = 5n$), sparse graphs ($m = 50n$), dense graphs ($m = n^{1.5}$), and very dense graphs ($m = n^2$). We also show similar results for \textsc{Par-Mod} in the \intheapp{}.} We see that for both of our algorithms and across different resolutions ($\lambda = 0.01, 0.85$), the running times of our algorithms scale nearly linearly with the number of edges.
{\color{change} Figure \ref{fig:exp6} shows the speedups of \textsc{Par-CC} on amazon, orkut, twitter, and friendster, over different numbers of threads. We also show the speedups of \textsc{Par-Mod} on the same graphs is given in the \intheapp{}. Overall, we see good parallel scalability, with 5.59--14.97x self-relative speedups on \textsc{Par-CC}, and 1.89--14.51x
self-relative speedups on \textsc{Par-Mod}. 
Note that using fewer threads for \textsc{Par-CC} times out on friendster for $\lambda = 0.01$, and we again see lower speedups for \textsc{Par-Mod} on twitter, where there is increased contention in using atomic compare-and-swaps due to the very few clusters produced relative to the size of the graph. Excluding twitter, we see 5.29--14.51x self-relative speedups on \textsc{Par-Mod}.}

{\color{change} 
\myparagraph{Memory Usage} Figure \ref{fig:expmem} shows the memory usage of \textsc{Par-CC} and \textsc{Par-Mod} on amazon, orkut, twitter, and friendster. \footnote{The size of the input graph provided in Figure \ref{fig:expmem} is the total size in CSR format \cite{Saad03}, which uses approximately 8 bytes per undirected edge.} Theoretically, our memory usage is linear in the size of the input graph. We incur more memory when using multi-level refinement, particularly if more coarsening rounds are required. This is because with refinement, we store the intermediate coarsened graph from each recursive step, whereas without refinement, we discard these graphs. 
For instance, we see that for \textsc{Par-CC} on friendster with $\lambda = 0.01$, four coarsening rounds 
are used, compared to $\lambda = 0.85$, where only one coarsening round is used, hence the difference in memory overhead.
Overall, using refinement, we incur a 1.40--23.68x memory overhead over the size of the input graph, whereas without refinement, we incur a 1.25--3.24x overhead.}

\myparagraph{Comparisons to Other Implementations}
In the special case of modularity, we compare against the highly optimized parallel modularity clustering implementation \textsc{networkit} \cite{networkit}.
\textsc{networkit}, like \textsc{Par-Mod}, implements an asynchronous version of Louvain-based modularity clustering. We discuss in the \intheapp{} the speedups of \textsc{Par-Mod} over \textsc{networkit}. We see up to 3.50x speedups (1.89x on average), primarily due to our optimization of the graph compression step, which we discuss below. We also obtain between 0.99 -- 1.00x the modularity given by \textsc{networkit}'s implementation, where some variance appears due to the asynchronous nature of both implementations.

\begin{figure}[t!]
\centering
\begin{subfigure}[t]{0.21\textwidth}
\centering
\includegraphics[width=\textwidth,page=5]{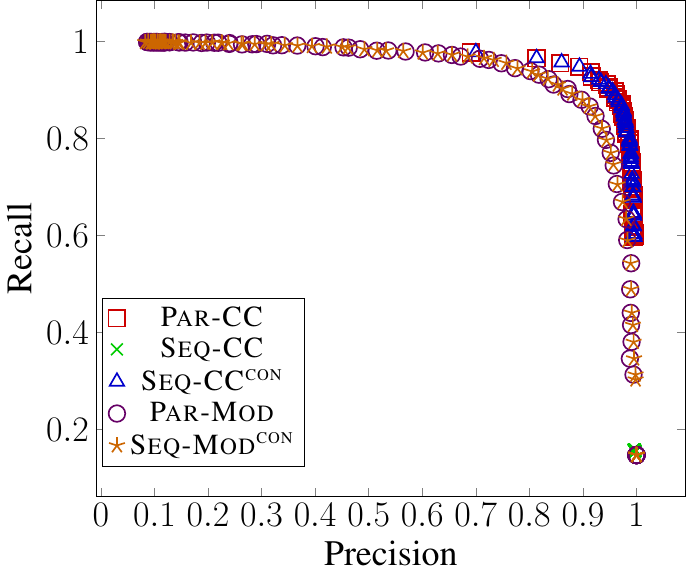}
\end{subfigure}
\hfill
\begin{subfigure}[t]{0.21\textwidth}
\centering
\includegraphics[width=\textwidth,page=6]{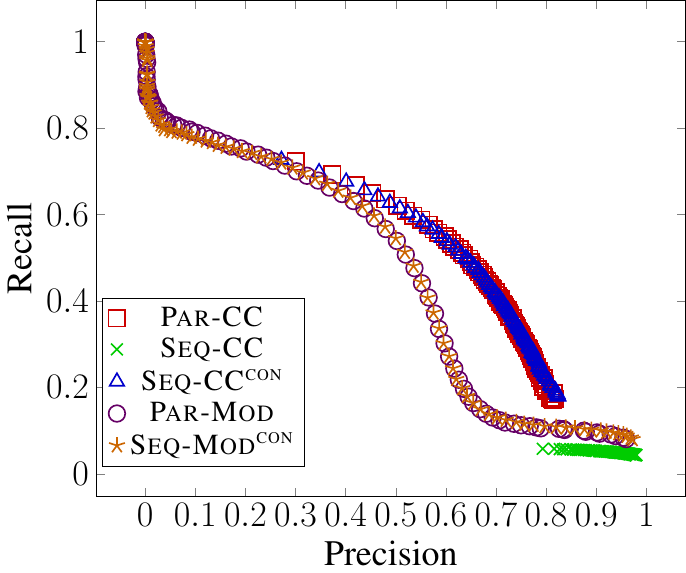}
\end{subfigure}
\caption{Average precision and recall compared to ground truth communities on amazon (left) and orkut (right) of the clusters obtained from \textsc{Par-CC} and \textsc{Par-Mod}, compared to \textsc{Seq-CC} and \textsc{Seq-Mod}, using varying resolutions.}
\label{fig:exp2-pr}
\end{figure}

\begin{figure}[t!]
\centering
\begin{subfigure}[t]{0.19\textwidth}
\centering
\includegraphics[width=\textwidth,page=17]{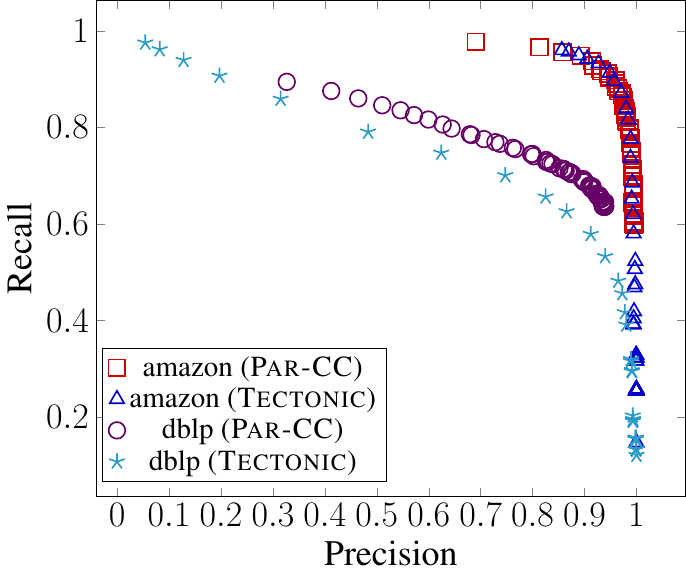}
\end{subfigure}
\hfill
\begin{subfigure}[t]{0.28\textwidth}
\centering
\includegraphics[width=\textwidth,page=20]{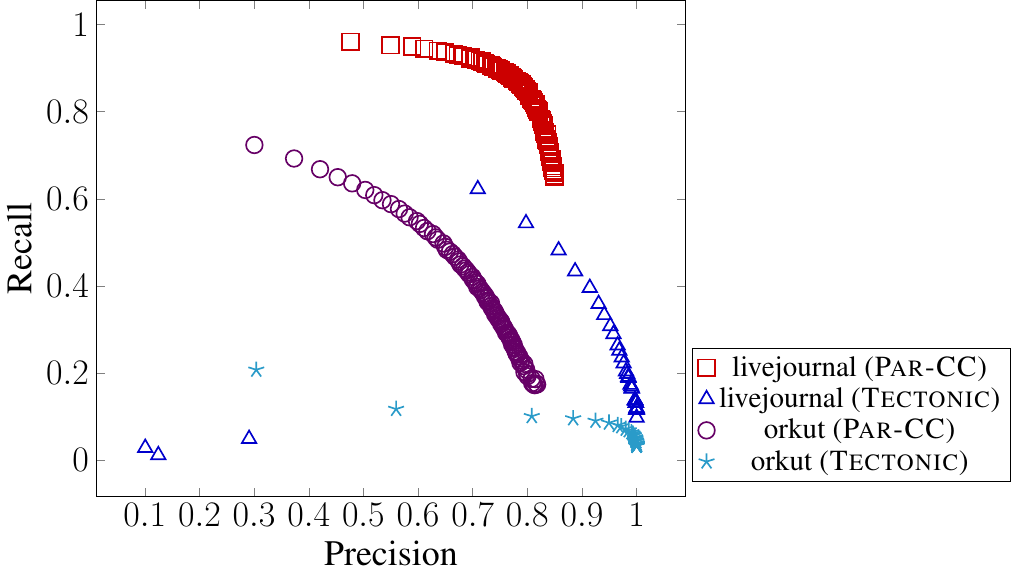}
\end{subfigure}
\caption{Average precision and recall compared to ground truth communities on amazon and dblp (left), and on livejournal and orkut (right), of the clusters obtained from \textsc{Par-CC} using varying resolutions, and from \textsc{Tectonic} using varying $\theta$.}
\label{fig:exp2-pr-overlay}
\end{figure}

Our speedups over \textsc{networkit} are primarily because \textsc{networkit} does not efficiently parallelize the graph compression step between rounds of best vertex moves. Our implementations use a work-efficient algorithm to parallelize this step, where intra-cluster edges are aggregated in polylogarithmic depth with an efficient parallel sort, whereas no such guarantee is made in \textsc{networkit}.

Additionally, we compare to the sequential \textsc{Tectonic} implementation \cite{TsPaMi17}, which clusters based on the idea of triangle conductance and provides good average precision-recall compared to ground truth communities on SNAP graphs. Like \textsc{Par-CC}, \textsc{Tectonic} uses a parameter $\theta$ that can be set to achieve a range of clusters with varying average precision and recall. We discuss in more detail in Section \ref{sec:exp:quality:uw} the comparison between the quality of \textsc{Par-CC}'s and \textsc{Tectonic}'s clusters, but for outputs where \textsc{Par-CC} either outperforms or matches \textsc{Tectonic} in terms of average precision and recall considering $\lambda \in \{ 0.01x \mid x \in [1, 99]\}$ and $\theta \in \{0.01x \mid x \in [1, 299]\}$ respectively, we see between 2.48--67.62x speedup of \textsc{Par-CC} over \textsc{Tectonic} on the graphs amazon, dblp, livejournal, and orkut. Notably, \textsc{Par-CC} significantly outperforms \textsc{Tectonic} on large graphs, with between 34.22--67.62x speedups on orkut.

Finally, we compare to \textsc{SCD} \cite{SCD}, a parallel triangle-based community detection implementation. \textsc{Par-CC} gives up to 2.89x speedups over \textsc{SCD} for resolutions that give comparable or better quality than \textsc{SCD}, in terms of average precision and recall. We discuss further details in the \intheapp{}.

\subsection{Quality Compared to Ground Truth} \label{sec:exp:quality:uw}

Figure \ref{fig:exp2-pr} shows the average precision-recall curves obtained by \textsc{Par-CC} and \textsc{Par-Mod}, varying resolutions, compared to the top 5000 ground truth communities on amazon and orkut. For \textsc{Par-CC}, we set $\lambda \in \{ 0.01 x \mid x \in [1, 99] \}$, and for \textsc{Par-Mod}, we set $\gamma \in \{ 0.02 \cdot (1.2)^x \mid x \in [1, 99]\}$.
We compare these curves to those obtained by \textsc{Seq-CC} and \textsc{Seq-Mod}, running with both the same number of iterations ($num\_iter=10$) as the parallel implementations, and to convergence. Note that we only show \textsc{Seq-Mod}${}^\textsc{con}$, because the version limiting the number of iterations displays precisely the same average precision-recall curve as the version running to convergence.

Overall, the average precision-recall obtained by \textsc{Par-CC} and \textsc{Par-Mod} match those obtained by their sequential counterparts for both amazon and orkut. Note that if we do not run \textsc{Seq-CC} to convergence, we obtain relatively poor precision-recall compared to \textsc{Par-CC}, suggesting that more progress is made in fewer iterations in our parallel implementation. {\color{change} This is likely inherent in the behavior in the asynchronous setting of our implementations, where the consistency guarantees are relaxed in a way such that vertices can more easily move better clusters.
}
Overall, \textsc{Par-CC} offers a better precision-recall trade-off compared to \textsc{Par-Mod}, which shows the benefits of using the \ccobj{} objective. We also show in the \intheapp{} the same behavior on dblp and livejournal.

Figure \ref{fig:exp2-pr-overlay} shows the average precision-recall curves for \textsc{Tectonic} \cite{TsPaMi17}, considering $\theta \in \{0.01x \mid x \in [1, 299]\}$, which we compare to \textsc{Par-CC} on amazon, dblp, livejournal, and orkut. We see that \textsc{Tectonic} achieves similar precision-recall trade-offs on amazon, but \textsc{Par-CC} obtains much better precision-recall on dblp, livejournal, and orkut. Notably, \textsc{Tectonic} degrades significantly on the larger graphs livejournal and orkut compared to \textsc{Par-CC}.

We also conducted quality experiments on weighted graphs, but we defer the details to the \intheapp{}.

\section{Conclusion}
We have designed and evaluated a comprehensive and scalable parallel clustering framework, which captures both correlation and modularity clustering. Our framework  offers settings with trade-offs between performance and quality.
We obtained significant speedups over existing state-of-the-art implementations that scale to large datasets of up to billions of edges.
Moreover, we showed that optimizing for correlation clustering objective results in higher-quality clusters with respect to ground truth data, compared to other methods used in highly-scalable clustering implementations. This shows the significance of the correlation clustering objective for community detection.
Finally, we proved the \pcomplete{}ness of the Louvain-like algorithms for parallel modularity and correlation clustering.

\begin{acks}
This research was supported by NSF Graduate Research Fellowship
\#1122374.
We thank Christian Sohler for insightful discussions on evaluating clustering quality.
\end{acks}

\balance

\bibliographystyle{ACM-Reference-Format}
\bibliography{ref}

\appendix 
\section{Sequential Louvain Method}

Algorithm \ref{alg:sequential-louvain} shows the pseudocode for 
the classic sequential Louvain method from Blondel \textit{et al.} \cite{BlGuLaLe08}, \textsc{Sequential-CC}.

Note that the computation necessary to determine the cluster that a vertex $v$ should move to given the objective function is omitted from the pseudocode for simplicity, but it can be efficiently performed by maintaining in every iteration the total vertex weight of each cluster in $C$. More precisely, if we denote the total vertex weight of a cluster $c$ by $K_c$, the change in objective of a vertex $v$ moving from its current cluster $c$ to a new cluster $c'$ (where $c \neq c'$) is given by
$$\left(\sum_{u \in c', (u, v) \in E} w_{uv} - \lambda k_v K_{c'} \right) - \left(\sum_{u \in c, (u, v) \in E} w_{uv} - \lambda k_v K_c+ \lambda k_v^2\right).$$
In other words, the change in objective depends solely on $K_c$, $K_{c'}$, and the weights of the edges from $v$ to its neighbors in $c$ and $c'$.

\begin{algorithm}[t!]
\small
\caption{Sequential Louvain method for correlation clustering}
\begin{algorithmic}[1]
\Procedure{Sequential-CC}{$G$, $k$}
\State $C \leftarrow$ singleton clusters for each $v \in V$ \label{c:line-init}
\Do
\State $\sigma \leftarrow$ random permutation of $V$
\For{each $v = \sigma(i)$} \label{c:line-best-move}
\State Move $v$ to the cluster in $C$ that maximizes the \ccobj{} objective
\EndFor \label{c:line-end-best-move}
\doWhile{the objective \ccobj{($C$)} has increased} \label{c:line-best-move-loop}
\If{no moves were made}
\State \Return $C$ \label{c:line-no-moves}
\EndIf
\State $G', k' \leftarrow$ \textsc{Sequential-Compress}($G$, $C$) \label{c:line-compress}
\State $C' \leftarrow$ \textsc{Sequential-CC}($G'$, $k'$)
\State \Return \textsc{Sequential-Flatten}($C$, $C'$) \label{c:line-flatten}
\EndProcedure
\end{algorithmic}
\label{alg:sequential-louvain}
\end{algorithm}

\begin{table}[t!]
\begin{center}
\footnotesize
\begin{tabular}[!t]{l|r|r}
\toprule
{Graphs} & Num. Vertices & Num. Edges \\
\midrule
{amazon  }         & 334,863       &925,872  \\
{dblp }           & 317,080        &1,049,866  \\
{livejournal  }   & 3,997,962  &34,681,189  \\
{orkut  }    & 3,072,441    &117,185,083  \\
{twitter  } & 41,652,231 &1,202,513,046 \\
friendster & 65,608,366      &1,806,067,135 \\
\end{tabular}
\end{center}
\caption{Sizes of graph inputs.} \label{table:sizes}
\end{table}

\section{Other Optimizations}
We make use of other practical optimizations in our parallel implementation of \textsc{Parallel-CC}. First, we use the theoretically efficient parallel primitives available in the Graph Based Benchmark Suite (GBBS) \cite{DhShTsBlSh20}. 
Overall, these primitives and the work-efficient scheduler provided in GBBS offer on average a 1.43x speedup over Intel's Parallel STL library \cite{BlAnDh20}.
Importantly, we use the \textsc{EdgeMap} primitive from GBBS to maintain the frontier of neighbors of moved vertices or of modified clusters in each step of \textsc{Best-Moves}. \textsc{EdgeMap} takes a vertex subset and applies a user-defined function to generate a new vertex subset -- in our case, generated from specified neighbors. The primitive switches between a sparse and a dense representation of the subset depending on size, and the implementation of \textsc{EdgeMap} similarly changes depending on the size of the input subset and the number of outgoing edges.

We also efficiently parallelize the sequential graph compression and cluster flattening subroutines, \textsc{Sequential-Compress} and \textsc{Sequential-Flatten} respectively. Flattening a given clustering $C$ to the clustering $C'$ from the coarsened graph can be parallelized by maintaining a set of cluster IDs for each vertex (given by the index in $[0, n]$ of the cluster in $C$), and assigning for each vertex, the cluster ID in $C'$ corresponding to the cluster containing its cluster ID in $C$. Moreover, we parallelize graph compression by aggregating in parallel the edges in the original graph by the cluster IDs of their endpoints, and using parallel reduces to combine edges whose endpoints correspond to the same cluster ID. 

Furthermore, in computing each vertex's desired cluster on Line \ref{pc:line-compute-move}, we make use of a parallel and a sequential subroutine, which we choose heuristically depending on the degree of the vertex. The change in objective for moving a vertex $v$ to other clusters can be efficiently parallelized by iterating through $v$'s neighbors in parallel and using a parallel hash table~\cite{gil91a}, from the GBBS implementation, to maintain the sum of edge weights to neighbors in the same cluster. However, for vertices of small degree and with large $V'$, the parallel overhead of maintaining such a hash table for each vertex is too costly. On the other hand, for vertices of large degree and with small $V'$, the parallel overhead is negligible compared to the improved depth in utilizing a parallel hash table. We use a fixed threshold to choose between using the sequential subroutine versus using the parallel subroutine.

\section{Additional Experiments}

Table \ref{table:sizes} shows the details of the SNAP graphs that we perform experiments upon.

Figure \ref{fig:exp1-num-rounds} shows the size of $V'$ comparing neighbors of clusters and neighbors of vertices as the subset $V'$, for \textsc{Par-CC} on amazon and orkut, using the synchronous setting and no refinement. In particular, for $\lambda = 0.85$ on amazon, the neighbors of vertices optimization obtains 1.61x speedup over all vertices, whereas the neighbors of clusters optimization obtains 1.16x speedup over all vertices. This is reflected in Figure \ref{fig:exp1-num-rounds}, where there is a significant difference in the size of $V'$ in this case.

We note that we see lower speedups using \textsc{Par-Mod} on twitter. 
We suspect that this is because twitter has a few vertices of particularly high degree (the maximum degree is 2,997,487, compared to the maximum degree in friendster of 5,214), and,
across all resolutions, \textsc{Par-Mod} and \textsc{Seq-Mod} produce significantly fewer clusters on twitter relative to the size of the graph, compared to other graphs. In particular, for modularity clustering, the average cluster size on twitter is $2100$--$2.08 \times 10^{7}$ across all resolutions,
whereas the average cluster size on friendster, a graph of similar size, is 1.11. As such, we see significantly increased contention on twitter.

Figure \ref{fig:exp7-dense} shows the scalability of \textsc{Par-Mod} over rMAT graphs of varying sizes, with very sparse graphs where $m = 5n$, sparse graphs where $m = 50n$, dense graphs where $m = n^{1.5}$, and very dense graphs where $m = n^2$.
Figure \ref{fig:exp6-large} shows the speedups of \textsc{Par-Mod} on amazon, orkut, twitter, and friendster, over different numbers of threads.

We show in Figure \ref{fig:exp2-mod-dblp-lj} the average precision and recall compared to ground truth communities on dblp and lj, of the clusters obtained by \textsc{Par-CC} and \textsc{Par-Mod} for varying resolutions. As discussed in Section \ref{sec:exp:quality:uw}, we see that overall, \textsc{Par-CC} offers a better precision-recall trade-off compared to \textsc{Par-Mod}.

\subsection{Experiments on Existing Baselines}

We note that there exist no prior scalable correlation clustering baselines that offer high quality in terms of objective. The existing implementations are the parallel \textsc{C4} and \textsc{ClusterWild!} \cite{PaPaOyReRaJo15}, which are based on a maximal independent set algorithm, and the sequential Louvain-based method in \textsc{LambdaCC} \cite{VeGlWi18}.

\begin{figure}[t!]
\centering
\begin{subfigure}[t]{0.21\textwidth}
\centering
\includegraphics[width=\textwidth,page=6]{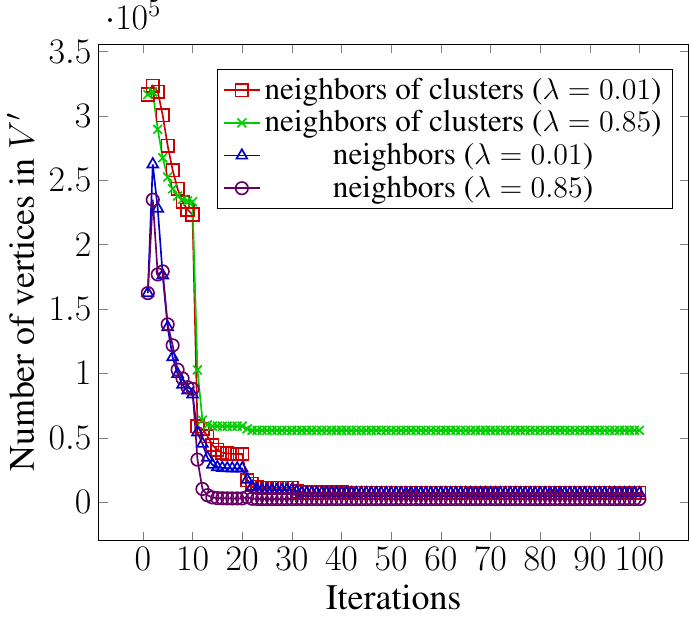}
\end{subfigure}
\hfill
\begin{subfigure}[t]{0.21\textwidth}
\centering
\includegraphics[width=\textwidth,page=7]{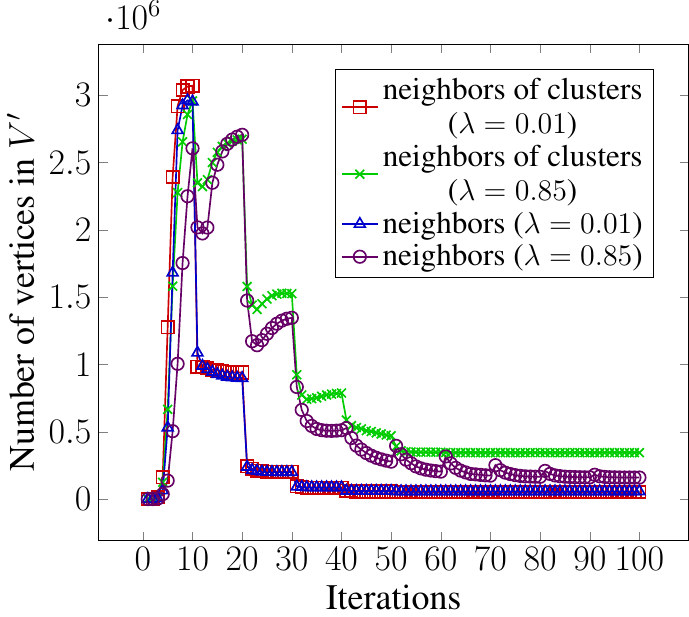}
\end{subfigure}
\caption{Number of vertices in $V'$ per best move iteration of \textsc{Par-CC}, for amazon (left) and orkut (right), considering neighbors of clusters and neighbors of vertices as the subset $V'$.}
\label{fig:exp1-num-rounds}
\end{figure}

\begin{figure}[t!]
\centering
\begin{subfigure}[t]{0.23\textwidth}
\centering
\includegraphics[width=\textwidth,page=14]{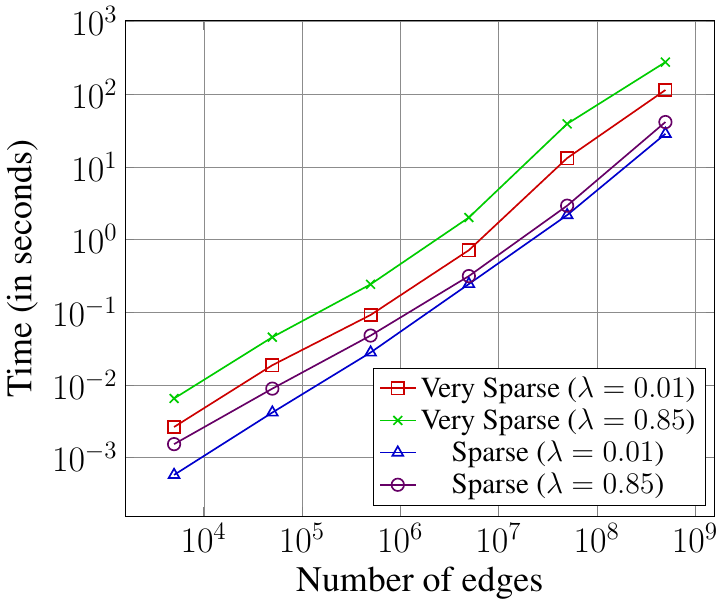}
\end{subfigure}
\hfill
\begin{subfigure}[t]{0.23\textwidth}
\centering
\includegraphics[width=\textwidth,page=9]{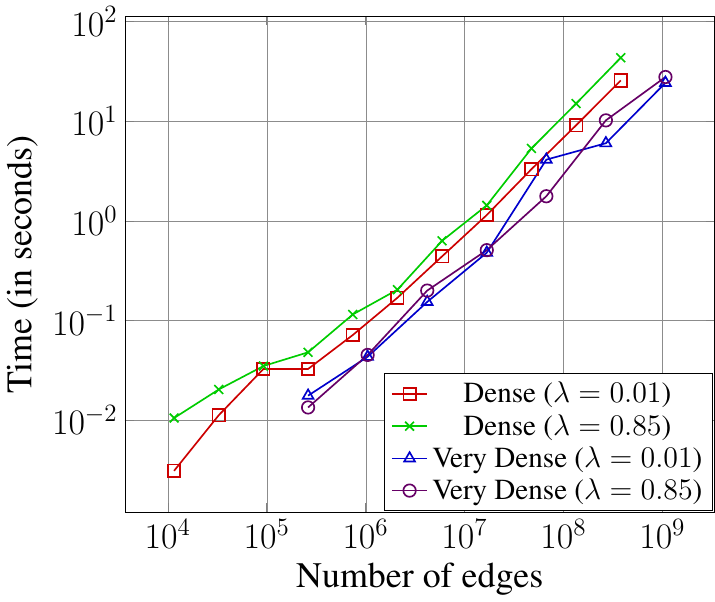}
\end{subfigure}
\caption{Scalability of \textsc{Par-Mod} over rMAT graphs with varying numbers of vertices.}
\label{fig:exp7-dense}
\end{figure}

\begin{figure}[t!]
\centering
\begin{subfigure}[t]{0.21\textwidth}
\centering
\includegraphics[width=\textwidth,page=4]{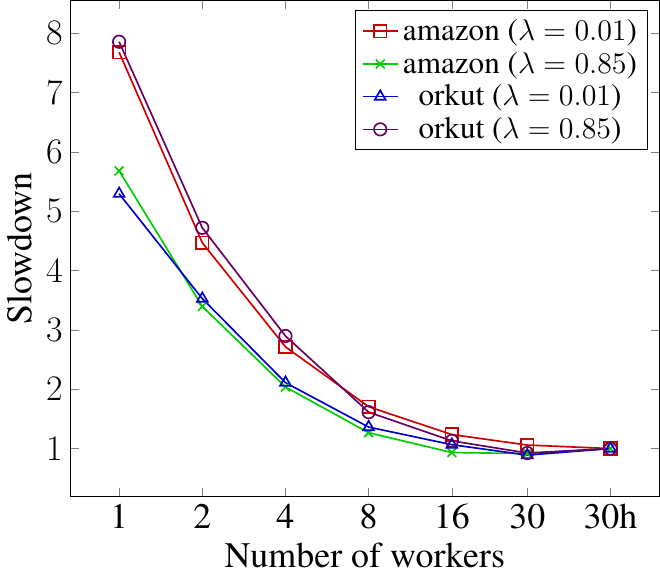}
\end{subfigure}
\hfill
\begin{subfigure}[t]{0.21\textwidth}
\centering
\includegraphics[width=\textwidth,page=2]{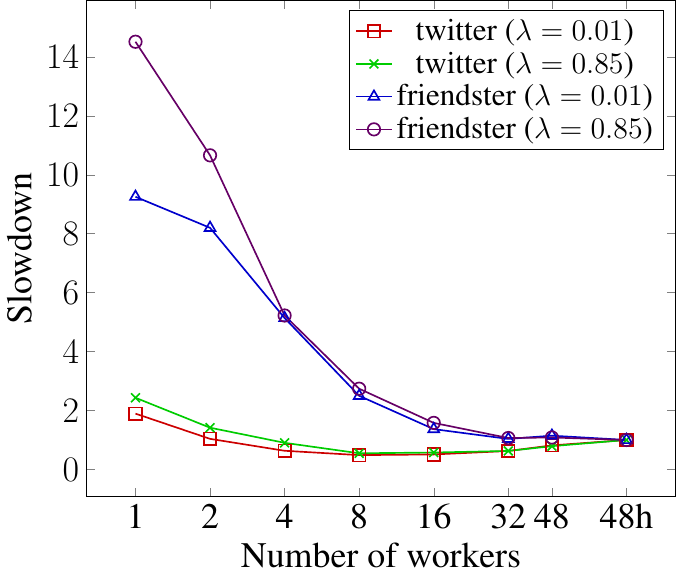}
\end{subfigure}
\caption{Scalability of \textsc{Par-Mod} over different numbers of threads, on amazon, orkut, twitter, and friendster, with $\lambda = 0.01, 0.85$. Note that 30h and 48h indicate 30 and 48 cores respectively, with two-way hyper-threading.}
\label{fig:exp6-large}
\end{figure}

The parallel correlation clustering implementations \textsc{C4} and \textsc{ClusterWild!} \cite{PaPaOyReRaJo15} optimize for the same objective as our \ccobj{} objective, if we set $\lambda = 0.5$; importantly, they do not generalize to other resolution parameters, and they do not take into account weighted graphs. We test the asynchronous versions of \textsc{C4} and \textsc{ClusterWild!}, which also outperform the synchronous versions while maintaining the objective, on the graphs amazon, dblp, livejournal, and orkut. Both implementations offer significant speedups over \textsc{Par-CC}, of up to 139.43x and 428.64x respectively. However, rescaling the objectives given by \textsc{C4} and \textsc{ClusterWild!} to match the \ccobj{} objective, we see that \textsc{C4} and \textsc{ClusterWild!} decrease the objective by 273.35 -- 433.31\% over \textsc{Par-CC}. Notably, the objectives given by \textsc{C4} and \textsc{ClusterWild!} are often negative, meaning that they are unsuitable for optimizing for the \ccobj{} objective. Moreover, compared to the top 5000 ground truth communities on these graphs, \textsc{C4} and \textsc{ClusterWild!} achieve poor precision and recall, with precision between 0.44--0.65 and corresponding recall between 0.10--0.15. In comparison, \textsc{Par-CC} achieves recall between 0.61--0.98 for precision greater than 0.50.

Additionally, we compare against the Louvain-based sequential correlation clustering implementation in \textsc{LambdaCC} given by Veldt \textit{et al.} \cite{VeGlWi18}. Unfortunately, this implementation does not scale to large graphs of more than hundreds of vertices. We were able to test \textsc{LambdaCC} on the karate graph \cite{Zachary77}, which consists of 34 vertices and 78 edges. For $\lambda = 0.01$, \textsc{LambdaCC} takes 0.057 seconds to cluster the karate graph, whereas our \textsc{Par-CC} takes 0.0002 seconds. The slowness of \textsc{LambdaCC} is because the code is in MATLAB, and it uses an adjacency matrix to represent the input graph; as such, it is unable to efficiently perform sparse graph operations.

We compare to \textsc{networkit} \cite{networkit} in the special case of modularity.
Note that \textsc{networkit}, like \textsc{Par-Mod}, implements an asynchronous version of Louvain-based modularity clustering, and requires a parameter $num\_iter$ to guarantee completion. By default \textsc{networkit} sets $num\_iter = 32$, so to compare with  \textsc{Par-Mod}, we similarly set $num\_iter = 32$.
We show in Figure \ref{fig:exp4-nk} the speedups of \textsc{Par-Mod} over \textsc{networkit} on amazon, dblp, livejournal, and orkut, for varying resolutions. We see up to 3.50x speedups, primarily due to our optimization of the graph compression step, and on average 1.89x speedups. For the twitter graph, setting $\gamma = 0.01, 0.85$, \textsc{Par-Mod} gives between 1.08 -- 3.03x speedups over \textsc{networkit} while maintaining comparable modularity. For the friendster graph, we turn off \textsc{networkit}'s turbo parameter (which offers a trade-off between memory usage and performance) due to space constraints, and setting $\gamma = 0.01, 0.85$, \textsc{Par-Mod} gives between 1.23--1.26x speedups over \textsc{networkit} while maintaining comparable modularity.

Finally, we compare to \textsc{SCD} \cite{SCD}, a parallel triangle-based community detection implementation. Note that \textsc{SCD} is not able to vary parameters to obtain clusters with significantly different average precision and recall. For amazon, dblp, and livejournal, our \textsc{Par-CC} implementation achieves 2.00--2.89x speedups over \textsc{SCD} while maintaining the same average precision and recall. For orkut, \textsc{SCD} obtains an average precision of 0.15 and an average recall of 0.05, while \textsc{Par-CC} can obtain an average precision of 0.61 and an average recall of 0.53 with 1.31x speedup over \textsc{SCD}.

\subsection{Experiments on Weighted Graphs}

We use an approximate $k$-NN algorithm to construct weighted graphs from pointset data. Specifically, we use the Optical Recognition of Handwritten Digits (digits) dataset (1,797 instances) and the Letter Recognition (letter) dataset (20,000 instances) from the UCI Machine Learning repository \cite{DuGr17}, both of which also have ground truth clusters which we compare to. We use the state-of-the-art ScaNN $k$-NN library~\cite{scann} to perform $k$-NN, with $k = 50$ and using cosine similarity. We symmetrize the resulting $k$-NN graph.

For weighted graphs $G$, we test both our implementations treating $G$ as an unweighted graph (with unit weight edges), and our implementations treating $G$ as a weighted graph. We denote the former with no superscript, and the latter with the superscript ${}^\textsc{w}$.
The sequential implementations \textsc{Seq-CC} and \textsc{Seq-Mod} give similar results to the corresponding parallel implementations, so we discuss only the parallel implementations.

\begin{figure}[t!]
\begin{center}
\includegraphics[width=0.45\columnwidth,page=21]{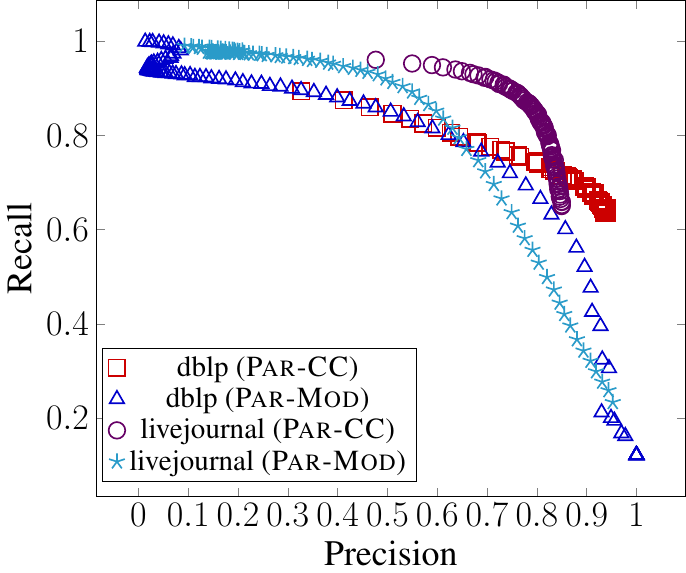}
\caption{Average precision and recall compared to ground truth communities on dblp and livejournal of the clusters obtained from our parallel implementations \textsc{Par-CC} and \textsc{Par-Mod}, using varying resolutions.}
\label{fig:exp2-mod-dblp-lj}
\end{center}
\end{figure}

\begin{figure}[t!]
\centering
\begin{subfigure}[t]{0.23\textwidth}
\centering
\includegraphics[width=\textwidth,page=7]{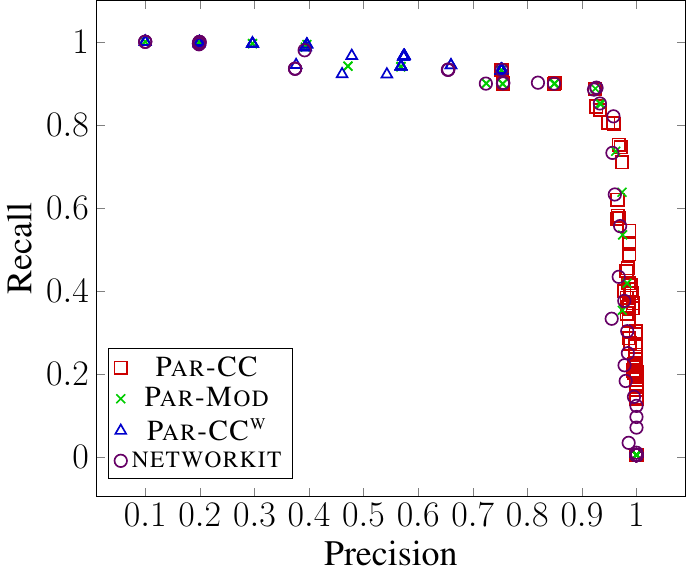}
\end{subfigure}
\hfill
\begin{subfigure}[t]{0.235\textwidth}
\centering
\includegraphics[width=\textwidth,page=10]{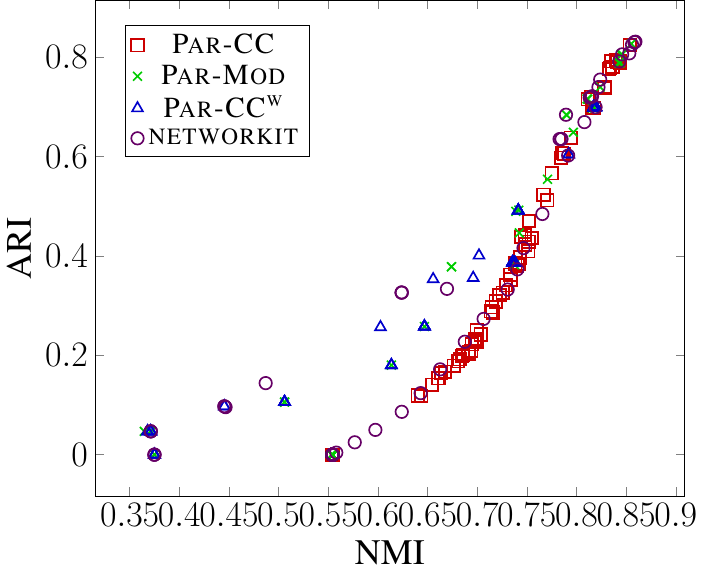}
\end{subfigure}
\caption{Average precision-recall and ARI-NMI scores for the digits graph, from \textsc{Par-CC}, \textsc{Par-Mod}, \textsc{Par-CC}${}^\textsc{w}$, and \textsc{networkit}.}
\label{fig:exp8-digits}
\end{figure}

\begin{figure}[t!]
\centering
\begin{subfigure}[t]{0.23\textwidth}
\centering
\includegraphics[width=\textwidth,page=8]{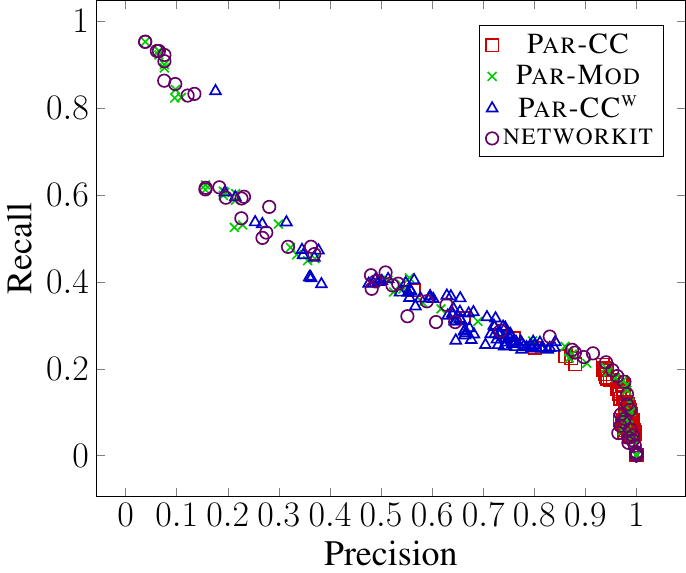}
\end{subfigure}
\hfill
\begin{subfigure}[t]{0.24\textwidth}
\centering
\includegraphics[width=\textwidth,page=9]{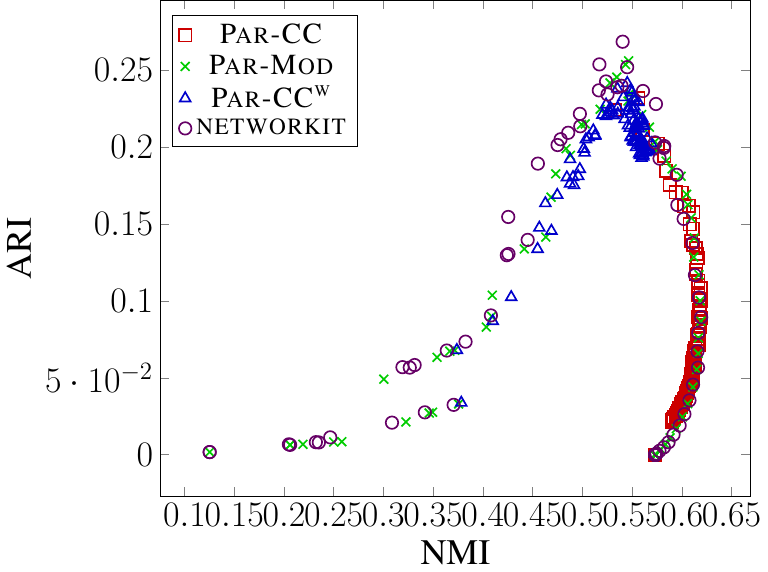}
\end{subfigure}
\caption{Average precision-recall and ARI-NMI scores for the letter graph, from \textsc{Par-CC}, \textsc{Par-Mod}, \textsc{Par-CC}${}^\textsc{w}$, and \textsc{networkit}.}
\label{fig:exp8-letter}
\end{figure}

Figures \ref{fig:exp8-digits} and \ref{fig:exp8-letter} show the average precision-recall and ARI-NMI scores for the digits and letter graph respectively. We compare our parallel implementations to \textsc{networkit}'s modularity clustering implementation, which can also take as input weighted graphs \cite{networkit}. \textsc{networkit} matches our \textsc{Par-Mod}${}^\textsc{w}$ implementation, so we omit the latter from our figures. We consider our implementations both treating the graphs as unweighted, and taking into account the edge weights. Moreover, we consider a range of resolutions, with $\lambda \in \{ 0.01 x \mid x \in [1, 99] \}$ for the \ccobj{} objective, and $\gamma \in \{ 0.02 \cdot (1.2)^x \mid x \in [1, 99]\}$ for the modularity objective.
Overall, compared to \textsc{networkit}, \textsc{Par-CC}${}^\textsc{w}$ is more robust across different resolution parameters compared to other implementations.

\section{P-Completeness of Louvain}

We prove here that the problem of obtaining the clustering given by the Louvain method maximizing for the \ccobj{} objective is P-complete.

\begin{theorem}\label{pcomplete}
The problem of obtaining a clustering equivalent to that given by the Louvain method maximizing for the \ccobj{} objective is P-complete.
\end{theorem}

\begin{proof}
 We set $\lambda = 0$ for the purposes of this proof. We will show that there is an NC reduction from the monotone circuit-value problem (CVP), where given an input circuit $C$ of $\wedge$ and $\lor$ gates on $n$ variables $(x_1, \ldots, x_n)$ and their negations, and an assignment of truth to these variables, the problem is to compute the output value of $C$.

The reduction works as follows. We use a fixed small $\varepsilon > 0$, and we construct a graph $G$. Initially, the vertices of $G$ are the literals $(x_1, \ldots, x_n)$ and their negations, and two additional vertices $t$ and $f$ (representing true and false respectively). We assign an edge with a large enough constant negative weight between $t$ and $f$. We also assign an edge with a large enough constant positive weight between each literal and its corresponding $t$ or $f$, depending on if the literal is true or false respectively.

Then, for every gate, say $g_i \lor g_j$ which outputs to $g_k$, we add a corresponding vertex $g_k$ and $g_k'$ in $G$. We use a weight $w_{ijk}$, which we define later. We also add edges of weight $w_{ijk}$ between $(g_i, g_k)$ and $(g_j, g_k)$, and an edge of weight $(2 + \frac{2}{3} \cdot \varepsilon) w_{ijk}$ between $(g_k, g_k')$. Finally, we add an edge of weight $(1 + \varepsilon) w_{ijk}$ between $(g_k, t)$, and an edge of weight $(1 + \frac{1}{2} \cdot \varepsilon) w_{ijk}$ between $(g_k, f)$. Figure \ref{fig:pcomplete} shows these vertices and edges.

We then define $w_{ijk}$ as follows, for each gate $g_i \lor g_j$. We take the topological sort of the directed acyclic graph given by the circuit $C$, where the vertices correspond to the gates and the literals. We call the ordered vertices of this DAG given by the topological sort $(c_1, \ldots, c_n)$, in order, and we let $c(g)$ denote the $c_i$ that corresponds to the gate or literal $g$. Then, we define a function $f(c_i) = \slfrac{1}{\prod_{1 \leq j < i} d_{c_j}}$, where $d_{c_j}$ denotes the degree of $c_j$. We note that $f(\cdot)$ can be efficiently computed for each $c_i$ by taking a prefix product of the degrees of $c_i$, in order. Finally, given a gate $g_i \lor g_j$ which outputs to $g_k$, we set $w_{ijk} = \min(f(c(g_i)), f(c(g_j)))$.

Note that the construction for every $\wedge$ gate is performed similarly, except we swap the edge weights between $(g_k ,t)$ and $(g_k, f)$. More explicitly, we instead add an edge of weight $(1 + \frac{1}{2} \cdot \varepsilon) w_{ijk}$ between $(g_k, t)$, and an edge of weight $(1 + \varepsilon) w_{ijk}$ between $(g_k, f)$.

We now claim that applying the Louvain method optimizing for the \ccobj{} objective on this graph $G$ solves the circuit $C$. In other words, we claim that each gate $g_i$ will ultimately cluster with either $t$ or $f$, depending on if its value in the circuit is true or false respectively. Recall that the Louvain method involves two main subroutines that iterate until convergence -- a subroutine in which vertices move to their best local clusters, and a subroutine that compresses the graph. We actually prove that the clustering given by performing best local vertex moves until convergence on $G$ produces two final clusters of this nature, one containing $t$ and one containing $f$. Then, in the compression subroutine, we obtain two vertices in the compressed graph corresponding to $t$ and $f$, and because the edge weight between $t$ and $f$ is a large enough negative constant, the two vertices in the compressed graph will not cluster with each other, terminating the algorithm.

We begin by proving a weaker statement, namely that at any given point in time throughout the best local vertex moves process, each gate $g_i$ is clustered into either a) a singleton cluster containing only $g_i$, b) a two-vertex cluster containing $g_i$ and $g_i'$, or c) a cluster containing $g_i$ and either $t$ or $f$ (but not both), depending on $g_i$'s corresponding truth value.

We prove this statement using induction. The base case follows because we begin with singleton clusters. We also note that the literals $x_i$ and their negations will always choose to cluster with their corresponding $t$ or $f$, because of the large enough positive constant weight between the edge from $x_i$ to its corresponding truth value. Similarly, the vertices $t$ and $f$ will always choose to cluster with corresponding $x_i$, due to the large enough positive constant weight. As such, we disregard literals and the vertices $t$ and $f$ in our inductive step. For our inductive step, we assume the inductive hypothesis, and show that the statement holds when we consider the local best move of a vertex $g_k$, originating from a gate $g_i \lor g_j$ (the argument for a gate $g_i \wedge g_j$ follows symmetrically).

\begin{figure}[t!]
\begin{center}
\includegraphics[width=0.45\columnwidth,page=18]{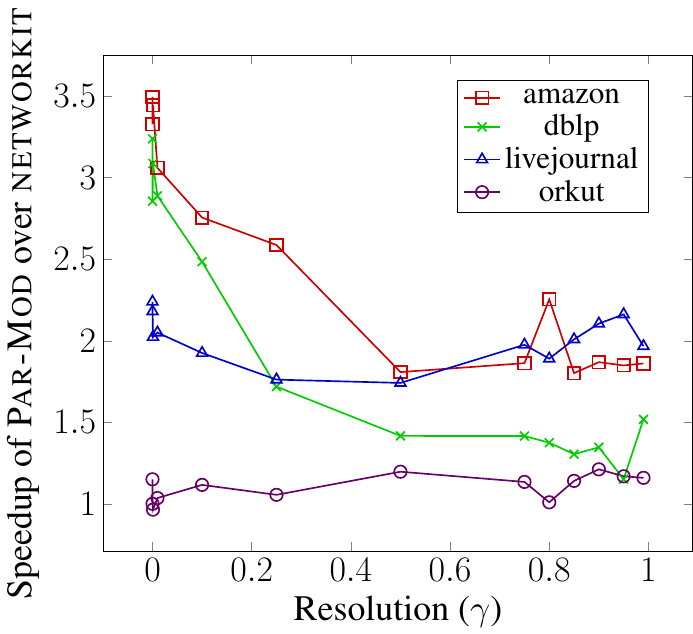}
\caption{Speedup of \textsc{Par-Mod} over \textsc{networkit} on amazon, dblp, livejournal, and orkut, for varying resolutions.}
\label{fig:exp4-nk}
\end{center}
\end{figure}

 \begin{figure}[t!]
\begin{center}
\includegraphics[width=0.7\columnwidth,page=2]{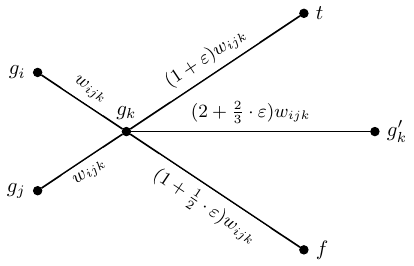}
\caption{The vertices and edges added to the graph $G$, given a gate $g_i \lor g_j$ which outputs to $g_k$.}
\label{fig:pcomplete}
\end{center}
\end{figure}

 Note that $g_k'$ only has a single positive weight edge to $g_k$, so its local best move is always to move to the cluster that $g_k$ is in. If neither $g_i$ nor $g_j$ are clustered with $t$ or $f$ when $g_k$ decides its best move, then $g_k$ will always choose to move to the cluster that $g_k'$ is in. This follows by construction because the edge weight $(g_k, g_k')$ of $(2 + \frac{2}{3} \cdot \varepsilon) w_{ijk}$ exceeds the edge weights from $g_k$ to $g_i$. $g_j$, $t$, and $f$, and because the sum of the weights of the edges from $g_k$ to its out-neighbors must be less than $w_{ijk}$. Moreover, since $g_k'$ only ever decides to cluster with $g_k$, either $g_k$ will remain in its current cluster, which satisfies the inductive step, or $g_k$ will move to cluster with $g_k'$, which must be a singleton cluster. In this case, $g_k$ forms a two-vertex cluster with $g_k'$.

 If at least one of $g_i$ and $g_j$ is clustered with $t$ when $g_k$ decides its best move, WLOG $g_i$, then $g_k$ will always choose to move to cluster with $t$, which by the inductive hypothesis must correspond with $g_k$'s truth value. This is because the sum of the weight of the edges $(g_k, t)$ and $(g_k, g_i)$ is given by $(2 + \varepsilon) w_{ijk}$, which exceeds the weight of the edge $(g_k, g_k')$, or $(2 + \frac{2}{3} \cdot \varepsilon) w_{ijk}$, thus inducing $g_k$ to move if it was originally in either a singleton cluster or a two-vertex cluster with $g_k'$. As before, the sum of the weights of the edges from $g_k$ to its out-neighbors must be less than $w_{ijk}$, so $g_k$ will not consider any other cluster.
 
Similarly, if both of $g_i$ and $g_j$ are clustered with $f$ when $g_k$ decides its best move, then $g_k$ will always choose to move to cluster with $f$, which by the inductive hypothesis must correspond with $g_k$'s truth value. This again follows directly from the weights of the edges from $g_k$ to $g_i$, $g_j$, $t$, $f$, and $g_k'$. Moreover, if exactly one of $g_i$ and $g_j$ is clustered with $f$, and neither are clustered with $t$, then $g_k$ will always choose to move to the cluster that $g_k'$ is in, by virtue of the constructed edge weights. In both of these cases, the inductive step is ultimately satisfied.

Thus, we have shown that at any given point in time throughout the best local vertex moves process, each gate $g_i$ is clustered into either a) a singleton cluster containing only $g_i$, b) a two-vertex cluster containing $g_i$ and $g_i'$, or c) a cluster containing $g_i$ and either $t$ or $f$ (but not both), depending on $g_i$'s corresponding truth value.

To complete the proof, we must now show that the best moves process converges with each vertex clustered with either $t$ or $f$, depending on its corresponding truth value. This follows from a similar argument to our inductive argument above, where for a gate $g_i \lor g_j$ that outputs to $g_k$, if $g_i$ and $g_k$ are both clustered with their corresponding truth value, then $g_k$ will always choose to cluster with its corresponding truth value during its best move operation. Furthermore, the literals $x_i$ and their negations necessarily choose to cluster with their corresponding truth values whenever prompted, by construction of the edge weights between the literals and $t$ and $f$. Thus, it follows that the best moves process converges with each gate $g_i$ clustered with either $t$ or $f$.

This completes the reduction, since we can obtain from the final clusters the solution to the circuit $C$.

\end{proof}

\end{document}